\documentclass[journal]{IEEEtran}

\usepackage{amsmath,amsfonts,amssymb}
\usepackage{epsfig}
\usepackage{latexsym}
\usepackage{mathrsfs}
\usepackage{multirow}
\usepackage{psfrag}
\usepackage{graphicx,subfigure}
\usepackage{graphicx}
\usepackage{cite}
\usepackage{colortbl}
\usepackage{booktabs}
\usepackage{blkarray}
\usepackage{stfloats}
\usepackage{arydshln}
\usepackage{extarrows}
\usepackage{algorithm}
\usepackage{algpseudocode}
\usepackage{color}
\usepackage{xcolor}
\usepackage{colortbl,booktabs}%
\usepackage{amsmath}
\usepackage{multicol,lipsum}
\usepackage{epstopdf}
\usepackage{float}
\usepackage{xcolor}
\usepackage{xpatch}
\usepackage{cuted}

\newtheorem{remark}{Remark}
\newtheorem{theorem}{Theorem}
\newtheorem{corollary}{Corollary}
\newtheorem{definition}{Definition}

\ifCLASSINFOpdf
\else
\fi
\begin{document}
\title{A Cooperative Implementation of Mesh Stability\\ in Vehicular Platoons
\thanks{This project has received funding from the National Natural Science Foundation of China grants Nos. 62073074, 62233004 and 62073076, the Jiangsu Provincial Key Lab of Networked Collective Intelligence grant BM2017002, and the European Union's Horizon 2020 research and innovation programme under the Marie Sklodowska-Curie grant agreement No 899987. {\it (Corresponding authors: Wenwu Yu and Simone Baldi).}
}}

\author{Meng Qiu,
        Di Liu,
        He Wang,
        Wenwu Yu,~\IEEEmembership{Senior Member,~IEEE}
        and~Simone Baldi,~\IEEEmembership{Senior Member,~IEEE}
        % <-this % stops a space
\thanks{M. Qiu, H. Wang, W. Yu and S. Baldi are with School of Mathematics, Frontiers Science Center for Mobile Information Communication and Security, Southeast University, Nanjing 210096, China (e-mails: \{230189527,101012806,wwyu\}@seu.edu.cn, s.baldi@tudelft.nl).}% <-this % stops a space
\thanks{D. Liu is with School of Computation, Information and Technology, Technical University of Munich (TUM), Germany, with Visual Intelligence for Transportation lab (VITA), Ecole Polytechnique Federale de Lausanne (EPFL), Switzerland, and also with School of Cyber Science and Engineering, Southeast University, Nanjing 210096, China (e-mail: di.liu@tum.de).}
% <-this % stops a space
%\thanks{He Wang is with the School of Mathematics, Southeast University, Nanjing 210096, China (e-mail:  101012806@seu.edu.cn.).}
%\thanks{Wenwu Yu is with the School of Mathematics, Southeast University, Nanjing 210096, China (e-mail: wwyu@seu.edu.cn.).}
%\thanks{S. Baldi is with School of Mathematics, Frontiers Science Center for Mobile Information Communication and Security, Southeast University, Nanjing 210096, China (e-mail: s.baldi@tudelft.nl).}
\thanks{Manuscript received xxx; revised xxx.}}


\markboth{Journal of \LaTeX\ Class Files,~Vol.~xx, No.~x, xx~xxxx}%
{Shell \MakeLowercase{\textit{et al.}}: Bare Demo of IEEEtran.cls for IEEE Journals}



\maketitle

\begin{abstract}
This work studies the problem of mesh stability in connected and automated vehicles. Mesh stability, also known as 2D string stability, refers to studying how disturbances propagate in vehicular platoons in both longitudinal and lateral direction. \textcolor{black}{As opposed to available decentralized results only relying on on-board sensing, the distinguishing feature of this work is a cooperative version of mesh stability, where on-board sensing is augmented by vehicle-to-vehicle communication.} This cooperative version dramatically improves the state-of-the-art decentralized performance: for longitudinal control, a new cooperative non-identical protocol (i.e. with non-identical control gains) is proposed that improves the state-of-the-art decentralized non-identical  protocol in terms of scalability of the control gains and strong notion of string stability. For lateral control, after showing that the non-identical approach is not necessary, a cooperative non-identical protocol is proposed to achieve another strong notion of string stability. Robustness of the proposed implementation against vehicle-to-vehicle communication delays \textcolor{black}{and actuation time lags} is considered. Numerical experiments, also performed with the vehicle simulator CarSim, validate the robustness and effectiveness of the proposed protocol.
\end{abstract}

\begin{IEEEkeywords}
Vehicular platoons, string stability, mesh stability, cooperative control, communication delay.
\end{IEEEkeywords}

% \fi
\IEEEpeerreviewmaketitle

\section{Introduction}
\IEEEPARstart{V}{ehicular} platooning refers to connected and automated vehicles (CAVs) driving at close distance from each other. The interest in this topic has steadily increased over the years as a way to efficiently cope with the increasing traffic flow \cite{12Naus}. Several studies have shown that vehicular platooning can improve fuel-consumption, traffic intelligence and safety \cite{a1,a2,28Bonnet}. Control of vehicular platooning involves both longitudinal and lateral directions: in the longitudinal direction, each vehicle should be controlled by braking or accelerating to maintain a desired inter-vehicle gap and velocity profile. Many researches on longitudinal control of platoons have been conducted, such as Cooperative Adaptive Cruise Control (CACC) \cite{13Ploeg, lin2020, Santini2019}, sliding mode control \textcolor{black}{\cite{tnse3,10Guo}}, reinforcement and deep learning \textcolor{black}{\cite{Wu2021,tnse1}} and nonlinear control \cite{15Besselink,Yue2020}.

\textcolor{black}{In a recent review} \cite{16Feng}, different platooning control methods have been categorized according to spacing policy, information network, disturbances type, and notion of string stability. The two most common inter-vehicle spacing policies are the constant-spacing policy \cite{14SWAROOP,19White} and the velocity-dependent policy \cite{10Guo,15Besselink}. Commonly considered information networks are the predecessor-follower (PF) and k-nearest neighbours (k-NN) \cite{a5}, which can be unidirectional (look-ahead) or bidirectional (look-ahead-and-behind) \cite{Liu,17Chehardoli,a6,04Xu,new6}. Information networks often involve delayed communication \cite{Santini2019,new3,new5,new4,new2,H} and/or security aspects \cite{Muthi2022,Pir2021}. Types of disturbance include non-zero initial conditions \cite{15Besselink} and external disturbances \cite{12Naus,10Guo}. In longitudinal platooning, the concept of string instability/stability refers to the capability of a controller to amplify/reduce disturbances (e.g. traffic shockwaves) as they propagate throughout the platoon:  %disturbance reduction can be utilized to absorb traffic shockwaves S
several notions of longitudinal string stability have been proposed in the literature, spanning from Lyapunov-based to frequency-domain characterizations \cite{Liu,17Chehardoli,a6,04Xu,18Khatir,29Papadimitriou,21Omar}. %Several notions of longitudinal string stability have been proposed in the literature, spanning from Lyapunov-based to frequency-domain characterizations: an overview of these notions can be found in \cite{16Feng}.

%Since the very beginning,
The analysis of string stability is traditionally performed in the longitudinal direction \cite{16Feng}. Although it is recognized in the literature that lateral string stability comes into play in curved paths, cornering and lane change \cite{19White,29Papadimitriou,21Omar,20Solyom}, lateral string stability is still an open research subject as compared to longitudinal string stability. The definition of $l_p$ and $l_\infty$ lateral string stability with look-down technique was first proposed in \cite{29Papadimitriou}; a special $l_2$ notion of lateral string stability is in \cite{20Solyom}, with the drawback that the lateral string stability conditions of the first three vehicles differ from the vehicles from the fourth to the last. By analyzing lateral string stability in a similar way as longitudinal frequency-domain methods, \cite{21Omar} considers $H_\infty$ control by approximating the road curvature with the preceding delayed information. Unfortunately, these works on lateral control and lateral string stability all neglect the effect of longitudinal control.

Combined longitudinal and lateral control of vehicle platoons can be classified into coupled \cite{01Bayuwindra,02Bayuwindra,03Zhao,04Xu} and decoupled \cite{18Khatir,05Wei,06ZHAO,a7}. The coupled approach requires some simplification steps due to the challenges of dealing with nonlinear dynamics: feedback linearization %control
based on unicycle model
is utilized to end up with linear dynamics easier to analyze. Unfortunately, coupled approaches %either
neglect %longitudinal and
lateral string stability, %or
and at most consider string stability only in the longitudinal direction \cite{04Xu,03Zhao}. The main appeal of decoupled approaches is to make use of robust control tools such as $\mathcal{L}_2$-gain control which can tackle uncertainty better than feedback linearization.
\textcolor{black}{Overall, few works have paid attention to keeping string stability jointly in longitudinal and lateral directions, a concept going under the name of mesh stability \cite{18Khatir,32Pant}. Among these, \cite{18Khatir} is particularly interesting since it showed that decentralized identical controllers for all vehicles do not achieve longitudinal string stability. Accordingly, \cite{18Khatir} proposes decentralized non-identical control, i.e. the controller is not the same for all vehicles.}
 \textcolor{black}{Although decentralized non-identical control can attain string stability, it only uses on-board measurements to keep the inter-vehicle distance. It is desirable to study non-identical control in a cooperative setting, i.e. where on-board sensing is augmented by vehicle-to-vehicle (V2V) communication.}
%require the gains to increase linearly with respect to the number of vehicles in the platoon, which can lead to high control efforts and oscillations.

In this work, \textcolor{black}{we show that this cooperative setting  dramatically improves the state-of-the-art decentralized performance in terms of scalability of the control gains and strong notion of string stability. It is worth mentioning that the V2V setting we study is in line with wireless communication in CACC technology, i.e. some predecessor's information is made available to the follower.} %global coordinates as most decoupled longitudinal and lateral literature, we consider relative motions of adjacent vehicles, which is more practical under V2V communication.
The main contributions are:
\begin{itemize}
\item %The proposed longitudinal controller takes advantage of
A novel  \emph{cooperative} non-identical longitudinal control is proposed. As compared to  \emph{decentralized} non-identical control ideas in the literature, we let the non-identical control gains increase less than linearly with the number of vehicles in the platoon, which is more scalable. Meanwhile, instead of weak string stability notions in the decentralized literature, we obtain disturbance rejection the sense of strong string stability. \textcolor{black}{Robustness to time lags approximating actuation delay is discussed};
\item It is shown that the non-identical approach is not necessary for lateral control. A novel cooperative  identical lateral controller is designed to achieve a suitable disturbance rejection performance, %only using the information of the preceding vehicle,
again in the sense of strong string stability;
\item In place of considering global coordinates as in most %decoupled
approaches, \textcolor{black}{relative coordinates are considered. As the relative coordinates describe the} relative motions of adjacent vehicles, \textcolor{black}{they better capture the on-board sensing of relative signals and the V2V communication among adjacent vehicles.} For both longitudinal and lateral cooperative control, robustness to delayed information due to V2V communication is discussed.
\end{itemize}

The remainder of the paper is organized as follows. In Section \ref{sec:02}, we %model the dynamics of the vehicle platoon and
give the preliminaries of string stability. Cooperative implementation of mesh stability is split into cooperative longitudinal control (Section \ref{sec:03}) and cooperative lateral control (Section \ref{sec:04}). In Section \ref{sec:05}, numerical tests are performed: the tests include a CarSim implementation, so as to validate the proposed approach in a realistic vehicle simulation platform. Conclusions are drawn in Section \ref{sec:06}.

%Throughout this paper,
{\bf Notations}: $ \mathbb{R}$ represents the set of real numbers; $j = \sqrt { - 1}$ is the imaginary unit for complex numbers; the $H_\infty$ norm of a transfer function is defined as ${\left\| G \right\|_{{\infty }}} = \mathop {\sup }\limits_\omega  | {G(j\omega )}|$; $\| \cdot \|$ stands for the Euclidean norm; $\hat x(s)$ indicates the Laplace transform of a signal $x(t)$, where $s$ is the Laplace operator.
\begin{figure}[tp]
    \centering
    \subfigure[][]{%
    \vspace{-0.1cm}
    \label{figure1}%
    \includegraphics[width=1.5in]{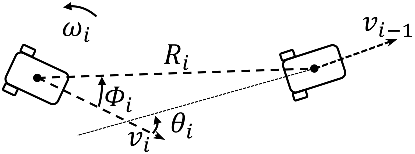}}
    \hspace{-40pt}%
    \hfill
    \subfigure[][]{%
    \label{figure1b}%
    \includegraphics[width=1.9in]{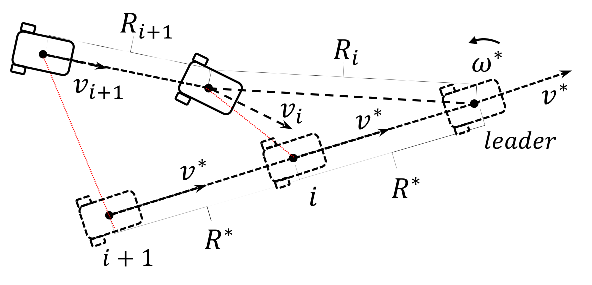}}
    \vspace{-0.1cm}
    \caption{Illustration of a vehicle platoon \ref{figure1} The reference dynamic unicycle model (edited from \cite{18Khatir}). \ref{figure1b} The desired position and orientation of the vehicle platoon shown in the dash line.}%
    \label{vehicle platoon}%
\end{figure}
%\vspace{-1pt}
\section{Model formulation and preliminaries}\label{sec:02}
\textcolor{black}{In order to consider curved driving and turning scenarios,} we adopt a unicycle kinematic vehicle model that combines longitudinal and lateral dynamics, cf. Fig. \ref{figure1}. This model is valid under the standard assumptions that roll, pitch and vertical motion are negligible \cite{18Khatir}. We consider a predecessor-follower topology where \textcolor{black}{relative position, relative velocity and relative orientation is acquired by the follower via on-board sensors (e.g. laser and camera), while} the predecessor's acceleration and angular velocity is acquired by the follower \textcolor{black}{via wireless V2V communication}. This is consistent with standard protocols for cooperative vehicles, such as CACC \cite{12Naus},\cite{13Ploeg},\cite{05Wei}. %For the analysis of string stability, we consider that the error dynamics are under zero initial conditions for both longitudinal and lateral direction. Under these assumptions,
In this setting, the vehicle dynamic is given by
\begin{equation}
\begin{aligned}\label{LL:01}
%\begin{cases}
\dot{R}_i(t)&=-v_i(t)\cos(\phi_i(t))+v_{i-1}(t)\cos(\theta_i(t)-\phi_i(t)),\\
\dot{v}_i(t)&=a_i(t),\\
\dot{\phi}_i(t)&=\frac{v_i(t)\sin(\phi_i(t))+v_{i-1}(t)\sin(\theta_i(t)-\phi_i(t))}{R_i(t)}-\omega_i(t),\\
\dot{\theta}_i(t)&=-\omega_i(t)+\omega_{i-1}(t),
%\end{cases}
\end{aligned}
\end{equation}
where $R_i$ is the relative distance between the center of gravity of vehicle $i$ and $i-1$, $\phi_i$ is the following angle between the velocity of vehicle $i$ and the relative distance vector, $\theta_i$ is the orientation angle and $v_i$ is the absolute linear velocity. The acceleration $a_i$ and angular velocity $\omega_i$ are the control inputs.

Following the decoupled approach, linearization can be made around the following equilibrium \cite{31Hu} of
\eqref{LL:01}
\begin{equation}\label{LL:02}
R_i=R^*,\ v_i=v^*,\ \phi_i=\phi^*,\ \theta_i=\theta^*,
\end{equation}
where $R^*$ and $v^*$ are the desired distance and velocity depending on the formation geometry, and $\phi^*$ and $\theta^*$ take into account the curvature of the path (tracking of a curved path may in general require a time-varying $\phi^*$ and $\theta^*$ \cite{02Bayuwindra}, \cite{05Wei}). The linearization of \eqref{LL:01} around the equilibrium \eqref{LL:02} results in:
\begin{subequations}\label{LL:03}
\begin{equation}
\begin{aligned}\label{LL:03a}
&\left[ \setlength{\arraycolsep}{2pt}{\begin{array}{*{20}{c}}
{{{\dot R}_i}(t)}\\
{{{\dot v}_i}(t)}
\end{array}} \right] = \left[ \setlength{\arraycolsep}{2pt}{\begin{array}{*{20}{c}}
0&{ - 1}\\
0&0
\end{array}} \right]\left[ \setlength{\arraycolsep}{1.5pt}{\begin{array}{*{20}{c}}
{{R_i}(t)}\\
{{v_i}(t)}
\end{array}} \right] +\left[ \setlength{\arraycolsep}{2pt}{\begin{array}{*{20}{c}}
0\\
1
\end{array}} \right]{a_i}(t) +\left[ \setlength{\arraycolsep}{2pt}{\begin{array}{*{20}{c}}
1\\
0
\end{array}} \right]{v_{i - 1}(t)},
\end{aligned}
\end{equation}
\begin{equation}
\begin{aligned}\label{LL:03b}
&\left[\setlength{\arraycolsep}{1.5pt} {\begin{array}{*{20}{c}}
{\begin{array}{*{20}{c}}
{{{\dot \phi }_i}(t)}
\end{array}}\\
{{{\dot \theta }_i}(t)}
\end{array}} \right] = \left[ \setlength{\arraycolsep}{1.8pt}{\begin{array}{*{20}{c}}
0&{\frac{{{v^*}}}{{{R^*}}}}\\
0&0
\end{array}} \right]\left[ \setlength{\arraycolsep}{1.8pt}{\begin{array}{*{20}{c}}
{{\phi _i}(t)}\\
{{\theta _i}(t)}
\end{array}} \right] +\left[ \setlength{\arraycolsep}{0.5pt}{\begin{array}{*{20}{c}}
{ - 1}\\
{ - 1}
\end{array}} \right]{\omega _i}(t) +\left[ \setlength{\arraycolsep}{1.8pt}{\begin{array}{*{20}{c}}
0\\
1
\end{array}} \right]{\omega _{i - 1}(t)},
\end{aligned}
\end{equation}
\end{subequations}
or, in a compact form
\begin{equation}\label{LL:04}
{\dot z_i}(t) = A{z_i}(t) + B{u_i}(t) + E{d_i}(t),
\end{equation}
with system states ${z_i}(t) = \begin{bmatrix}R_i(t) & v_i(t) & \phi_i(t) & \theta_i(t)\end{bmatrix}^T$, longitudinal and lateral control input ${u_i}(t) =\begin{bmatrix}a_i(t) & \omega_i(t) \end{bmatrix}^T$, exogenous input from the preceding vehicle ${d_i}(t) =\begin{bmatrix}v_{i-1}(t) & \omega_{i-1}(t)\end{bmatrix}^T$, and
$$A=\begin{bmatrix}
	0 &-1 & 0 & 0\\ 0 & 0 & 0 & 0 \\
	0 & 0 & 0& \frac{v^*}{R^*} \\
	0& 0& 0 & 0
\end{bmatrix},\ B=\begin{bmatrix}
0&0\\1&0\\0&-1\\0&-1
\end{bmatrix}, \ E=\begin{bmatrix}
1&0\\0&0\\0&0\\0&1
\end{bmatrix}.$$ All variables in \eqref{LL:03}-\eqref{LL:04} are to be intended as a deviation from the equilibrium \eqref{LL:02}.
We have that longitudinal and lateral analysis are decoupled into \eqref{LL:03a} and \eqref{LL:03b} respectively. %, which will be used for control and stability analysis.
These decoupled %longitudinal and lateral
dynamics are in line with standard literature \cite{05Wei,06ZHAO}.

The dynamics of the leading vehicle (i.e. the vehicle without any  predecessor) are described as
\begin{equation}\label{LL:00}
\begin{aligned}
\dot{v}_0(t)=a_0(t), \\
\dot{\theta}_0(t)=\omega_0(t),
\end{aligned}
\end{equation}
where $v_0$ and $\theta_0$ are the absolute linear velocity and orientation angle of the leader, while $a_0$ and $\omega_0$ are the leader's inputs. The leader's inputs play the role of external disturbances which should not be amplified throughout the vehicle platoon. Accordingly, we now define some errors to quantify the amplification of the disturbances.

\subsection{Error definitions}
The vehicle-following objective is to make the platoon drive with a certain orientation meanwhile maintaining constant distance and relative velocity between two adjacent vehicles, cf. Fig. \ref{figure1b}. For longitudinal direction, the distance error and velocity error of vehicle $i$ can be defined as
\begin{equation}\label{LL:06a}
\begin{aligned}
e_{i,1}(t)&=R_i(t)-R^*,  \\
e_{i,2}(t)&=v_{i-1}(t)-v_i(t).	
\end{aligned}
\end{equation}
%\begin{subequations}
%\begin{equation}\label{LL:06a}e_{i,1}=R_i-R^0,\end{equation}
%\begin{equation}\label{LL:06b}e_{i,2}=v_{i-1}-v_i.\end{equation}
%\end{subequations}
For lateral direction, the angle error and relative orientation error of vehicle $i$ are defined as
%\begin{subequations}
%\begin{equation}\label{LL:07a}e_{i,3}=\phi_i-\phi^0,\end{equation}
%\begin{equation}\label{LL:07b}e_{i,4}=\theta_i-\theta^0.\end{equation}
%\end{subequations}
\begin{equation}\label{LL:07a}
\begin{aligned}
e_{i,3}(t)&=\phi_i(t)-\phi^*,  \\
e_{i,4}(t)&=\theta_i(t)-\theta^*.	
\end{aligned}
\end{equation}
%Notice that for the lateral control, a constant desired state can be suitable for the turning or lane changing maneuver, whereas the tracking of a curved path in general requires a time-varying desired state \cite{02Bayuwindra}, \cite{05Wei}.

%\IEEEpubidadjcol

%\begin{figure}
%\begin{minipage}[t]{0.5\linewidth}
%\centering
%\includegraphics[width=1.98in]{2.pdf}
%%\caption{fig1}
%\end{minipage}%
%\begin{minipage}[t]{0.5\linewidth}
%\centering
%\includegraphics[width=1.98in]{1.pdf}
%%\caption{fig2}
%%\label{fig:side:b}
%\end{minipage}
%\caption{An example of ESS (left column) with $N=3$ in Definition \ref{def1} and SFSS (right column) edited from \cite{16Feng}, where $x_i(t)$ means the state fluctuation of vehicle $i$.}\label{99}
%\end{figure}

\subsection{Longitudinal and lateral string stability definitions}
The ability to attenuate a disturbance as it propagates throughout the platoon is called string stability \cite{13Ploeg}. The literature has proposed different ways to measure this effect. Popular definitions are the following.
\begin{definition}\label{def1}
\cite{16Feng} For a linear platoon system \eqref{LL:04} with predecessor-follower topology, the system is \emph{eventual string stable} (ESS), if the output transfer function between the leading vehicle $0$ and any other vehicle $i$, denoted as $G_{0,i}(s)$, satisfies that: there exists $N < m$, such that
\begin{equation}\label{LL:08}
{\left\| {{G_{0,i}(s)}} \right\|_{{\infty }}} \le 1, \ \forall i > N,\ m \in \mathbb{R}.
\end{equation}
\end{definition}
\begin{definition}\label{def2}
\cite{16Feng} For a linear platoon system \eqref{LL:04} with predecessor-follower topology, the system is \emph{strong frequency-domain string stable} (SFSS), if the output transfer function between vehicle $i-1$ and vehicle $i$, denoted as $G_{i-1,i}(s)$, satisfies that:
\begin{equation}\label{LL:09}
 \lVert G_{i-1,i}(s)\rVert_\infty \leq1, \ \forall i.
\end{equation}
\end{definition}
To further clarify the relationship between ESS and SFSS, note that SFSS implies ESS but not vice versa. In fact, SFSS is defined among any two adjacent vehicles independently on the platoon size, whereas ESS depends on the platoon size via the relation $N<m$. In other words, SFSS in \eqref{LL:09} requires each vehicle to attenuate disturbance, whereas ESS in \eqref{LL:08} allows disturbance amplifications inside the platoon, provided they are eventually rejected when reaching the end of the platoon. For this reason, ESS is also referred to as ``weak'' string stability \cite{16Feng}. Obviously, vehicular platoons with SFSS property outperform platoons with ESS property in terms of disturbance attenuation. In this work, we focus on SFSS.

For the special condition of $G_{i-1,i}(j\omega)=1, \forall \omega \geq 0$, the system falls into a non-robust degenerate scenario, called marginal string stability \cite{12Naus}, which is typically not desired. In most practical scenarios we have $G_{i-1,i}(0 )=1$ at zero frequency, and $\left|G_{i-1,i}(j\omega )\right|<1$ for $\omega>0$. In fact, the output transfer function $G_{i-1,i}$ can be regarded as the complementary sensitivity function (cf. \eqref{LL:15}, \eqref{LL:20} and \eqref{LL:37} later on), and $G_{i-1,i}(0)=1$ represents the possibility of vehicle $i$ to track a constant signal from vehicle $i-1$ (as expected in most complementary sensitivity functions \cite{12Naus} \cite{13Ploeg}). Therefore, based on the tracking errors \eqref{LL:06a}-\eqref{LL:07a}, $G_{i-1,i}$ takes the form
\begin{equation}\label{LL:10}
{\hat e_{i,j}}(s)=G_{i-1,i}(s){{\hat e_{i - 1,j}}(s)}, j\in\left\{ {1,2,3,4} \right\}, \ i=2,\cdots, n.
\end{equation}
\begin{remark}[Mesh stability] Definitions \ref{def1}-\ref{def2} have been originally proposed for longitudinal string stability \cite{13Ploeg,17Chehardoli,18Khatir}, although some literature used similar notions for disturbance propagation in the lateral direction \cite{29Papadimitriou}. %For example, lateral string stability of the linear dynamics was investigated  in \cite{29Papadimitriou} in the $l_\infty$ sense.
The work \cite{32Pant} was the first to propose a combined notion of longitudinal and lateral string stability, called mesh stability: a mesh stable platoon attenuates disturbances in both directions simultaneously. Sect. \ref{sec:03} and Sect. \ref{sec:04} will consider string stability in longitudinal and lateral sense, in accordance with mesh stability. %definition.
\end{remark}
\section{Longitudinal control}\label{sec:03}
This paper takes advantage of a feedforward and feedback information for designing \textcolor{black}{the input $a_i(t)$}, according to
\begin{equation}\label{LL:11}
a_{i}(t)=\alpha_i e_{i,1}(t)+\beta_i e_{i,2}(t)+\gamma_i a_{i-1}(t),
\end{equation}
 \textcolor{black}{where $\alpha_i$, $\beta_i$, $\gamma_i$ are control gains to be designed.} The terms "feedforward" and "feedback" underline that $e_{i,1}$ and $e_{i,2}$ are obtained via feedback from onboard sensors (e.g. laser), whereas the feedforward term $a_{i-1}$ is obtained by V2V communication as in CACC \cite{12Naus},\cite{13Ploeg}. The block diagram of the controlled system is in Fig. \ref{figure a}. First, we briefly prove that identical control gains $\alpha_i=\alpha$, $\beta_i=\beta$, $\gamma_i=\gamma$ for all vehicles, lead to a degenerate marginal string stability scenario. % with respect to achieving longitudinal SFSS as in Definition \ref{def2}.
\begin{figure}[t]
%\begin{center}
\centering
%\vspace{-10pt}
\scalebox{0.55}[0.55]{\includegraphics{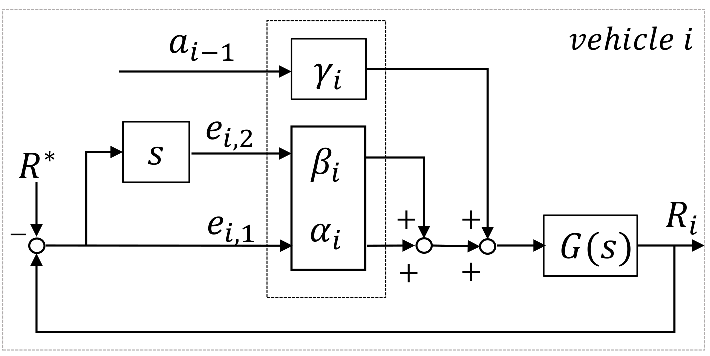}}
\caption{ Block diagram of %the feedforward and feedback
cooperative  non-identical control.} \label{figure a}
%\end{center}
\end{figure}

\subsection{Identical control approach}\label{sec:03a}
Consider the longitudinal error dynamics of \eqref{LL:06a}
\begin{equation}\label{LL:12}
\begin{aligned}
\dot{e}_{i,1}(t)&=e_{i,2}(t),\\
\dot{e}_{i,2}(t)&=a_{i-1}(t)-a_{i}(t),
\end{aligned}
\end{equation}
and substitute the controller \eqref{LL:11} using identical control gains into \eqref{LL:12}, so as to obtain
\begin{equation}\label{LL:13}
\ddot{e}_{i,1}(t)=-\beta\dot{e}_{i,1}(t)-\alpha e_{i,1}(t)+(1-\gamma)a_{i-1}(t).
\end{equation}
The transfer function of dynamics \eqref{LL:12} yields
\begin{equation}\label{LL:14}
\frac{{{{\hat e}_{i,1}(s)}}}{{{{\hat a}_{i-1}}(s)}} = \frac{{1 - \gamma }}{{ {s^2} + \beta s + \alpha }}.
\end{equation}
From the Routh-Hurwitz criterion, \eqref{LL:13} is asymptotically stable if and only if $\alpha, \beta>0$. After straightforward manipulations using the Laplace operator, one can obtain the following complementary sensitivity function $G_{i-1,i}(s)$ %for longitudinal string stability
according to \eqref{LL:10}
\begin{equation}\label{LL:15}
G_{i-1,i}(s)=\frac{{\hat e_{i,1}}(s)}{{\hat e_{i - 1,1}}(s)} = \frac{{\gamma {s^2} + \beta s + \alpha }}{{{s^2} + \beta s + \alpha }}.
\end{equation}
In \eqref{LL:15}, the SFSS condition \eqref{LL:09} can be translated as
\begin{equation}\label{LL:16}
\left| {\frac{{\alpha  - \gamma {\omega ^2} + \beta j\omega }}{{\alpha  - {\omega ^2} + \beta j\omega }}} \right| \le 1,\forall \omega
\end{equation}
which further implies
\begin{equation}\label{LL:17}
{\omega ^2}({\gamma ^2} - 1) + 2\alpha (1 - \gamma ) \le 0.
\end{equation}
It is clear that \eqref{LL:17} requires $-1 \le \gamma\leq1$ and $\gamma \geq 1$ simultaneously. The only possible value is $\gamma = 1$ which leads to the transfer function in \eqref{LL:15} to be exactly $1$. This is clearly a non-robust degenerate marginal string stability scenario. % with respect to string stability.

\subsection{Proposed non-identical control approach}\label{sec:03b}
\textcolor{black}{To overcome the non-robust degenerate scenario of Sect. \ref{sec:03a}, we propose the }controller \eqref{LL:11} with non-identical gains, for which the following theorem holds.
\begin{theorem}\label{thm1}
 [Longitudinal stability and string stability] Consider the longitudinal dynamics in \eqref{LL:03a}, with the non-identical cooperative controller \eqref{LL:11}. Then, the tracking errors \eqref{LL:06a} converge asymptotically when the control gains in \eqref{LL:11} are all positive. Moreover, the system is longitudinal SFSS if the control gains also satisfy the following conditions
\end{theorem}
\begin{equation}\label{LL:18}
\left\{ {\begin{aligned}
&{\gamma _i} \le \frac{1}{\gamma_{i-1}}, {\gamma _i}+\frac{1}{\gamma_{i-1}} \ge 2,\\
&{{\alpha _{i - 1}} \leq \frac{{1 - {\gamma _{i - 1}}}}{{1 - {\gamma _i}}}{\alpha _i}},\\
&{\beta _{i - 1}^2 \leq {{\left( {\frac{{1 - {\gamma _{i - 1}}}}{{1 - {\gamma _i}}}} \right)}^2}\left( { - 2{\alpha _i} + \beta _i^2} \right) + 2{\alpha _{i - 1}}{\gamma _{i - 1}}}.\\
\end{aligned}} \right.
\end{equation}

\begin{IEEEproof}
Similarly to \eqref{LL:14}, the stability of the error dynamics requires the control gains to be all positive. To pursue string stability with the non-identical cooperative feedback controller in \eqref{LL:11} and error dynamics in \eqref{LL:12}, one has
\begin{equation}\label{LL:19}
\frac{{{{\hat v}_i}(s)}}{{{{\hat v}_{i - 1}(s)}}} = \frac{{{\gamma _i}{s^2} + {\beta _i}s + {\alpha _i}}}{{{s^2} + {\beta _i}s + {\alpha _i}}}.
\end{equation}
After substituting \eqref{LL:19} into \eqref{LL:14}, the magnitude of the complementary sensitivity function $G_{i-1,i}(s)$ is obtained
\begin{align}\label{LL:20}
\left|G_{i-1,i}(s)\right|= &\left| {\frac{{{{\hat e}_{i,1}(s)}}}{{{{\hat e}_{i - 1,1}(s)}}}} \right|, \nonumber\\ = &\left| {\frac{{{{\hat e}_{i,1}(s)}}}{{{{\hat a}_i}(s)}}\frac{{{{\hat v}_i}(s)}}{{{{\hat v}_{i - 1}(s)}}}\frac{{{{\hat a}_{i - 1}(s)}}}{{{{\hat e}_{i - 1,1}(s)}}}} \right|,\nonumber\\
=&\left| {\frac{{1 - {\gamma _i}}}{{1 - {\gamma _{i - 1}}}}\frac{{{\gamma _{i - 1}}{s^2} + {\beta _{i - 1}}s + {\alpha _{i - 1}}}}{{{s^2} + {\beta _i}s + {\alpha _i}}}} \right|.
\end{align}
Condition \eqref{LL:09} applied to \eqref{LL:20} becomes a polynomial of $\omega$ as
\begin{align}\label{LL:21}
&\omega^4[(1-\gamma_i)^2\gamma^2_{i-1}-(1-\gamma_{i-1})^2]+
	\omega^2[-2(1-\gamma_i)^2\alpha_{i-1}\gamma_{i-1}\nonumber \\
&\quad+(1 -\gamma_i)^2\beta^2_{i-1}+2(1-\gamma_{i-1})^2\alpha_i-(1-\gamma_{i-1})^2\beta_i^2] \nonumber\\
&\quad +(1-\gamma_i)^2\alpha^2_{i-1}-(1-\gamma_{i-1})^2\alpha_i^2\leq 0, \quad \forall\omega.
\end{align}
%\begin{align}\label{LL:21}
%& {\omega ^4}[{{{( {1 - {\gamma _i}} )}^2}{\gamma ^2_{i - 1}} - {{\left( {1 - {\gamma _{i - 1}}} \right)}^2}}]+\nonumber\\
%&{\omega ^2}[- 2{( {1 - {\gamma _i}})^2}{\alpha _{i - 1}}{\gamma _{i - 1}}+ {\left( {1 - {\gamma _i}} \right)^2}{\beta ^2_{i - 1}} +\nonumber\\
%&\quad \quad \quad \quad \quad \quad  2{\left( {1 - {\gamma _{i - 1}}} \right)^2}{\alpha _i} - {\left( {1 - {\gamma _{i - 1}}} \right)^2}{\beta ^2_i}]+\nonumber\\
%&  {\left( {1 - {\gamma _i}} \right)^2}{\alpha ^2_{i - 1}} - {\left( {1 - {\gamma _{i - 1}}} \right)^2}{\alpha^2 _i}<0, \forall\omega.
%\end{align}
For \eqref{LL:21} to hold, it is sufficient that the coefficients of the polynomial \eqref{LL:21} are negative, which is true if the control gains are positive and satisfy the following condition
%\begin{gather}
% {{{\left( {1 - {\gamma _i}} \right)}^2}{\gamma^2 _{i - 1}} - {{\left( {1 - {\gamma _{i - 1}}} \right)}^2}}<0 \Rightarrow {\gamma _{i - 1}} < {\gamma _i} < 1,\nonumber\\
%{\left( {1 - {\gamma _i}} \right)^2}{\alpha ^2_{i - 1}} - {\left( {1 - {\gamma _{i - 1}}} \right)^2}{\alpha^2 _i}<0,\nonumber\\
%{\left( {1 - {\gamma _i}} \right)^2}\left( { - 2{\alpha _{i - 1}}{\gamma _{i - 1}} + {\beta ^2_{i - 1}}} \right) < {\left( {1 - {\gamma _{i - 1}}} \right)^2}\left( { - 2{\alpha _i} + {\beta^2 _i}} \right).\label{LL:22}
%\end{gather}
%(1-\gamma_i)^2\gamma^2_{i-1}-(1-\gamma_{i-1})^2<0 \Rightarrow
%\gamma_{i-1}<\gamma_i< 1,\nonumber\\
\begin{align}\label{LL:22}
&( {1 - {\gamma _i}} ){\gamma _{i - 1}} \le 1 - {\gamma _{i - 1}}, \nonumber\\
&(1-\gamma_i)^2\alpha^2_{i-1}-(1-\gamma_{i-1})^2\alpha^2_i \leq0,\\
&(1-\gamma_i)^2(-2\alpha_{i-1}\gamma_{i-1}+\beta^2_{i-1})\leq(1-\gamma_{i-1})^2(-2\alpha_i+\beta^2_i).\nonumber
\end{align}
The conditions in \eqref{LL:18} can be obtained from \eqref{LL:22} by simple manipulations. This concludes the proof.
\end{IEEEproof}

\begin{remark}
[Increase of gains] Non-identical control for longitudinal string stability was proposed in \cite{18Khatir} in the PID form
\begin{align}
u_i(s)=( {{P_i} + \frac{{{I_i}}}{s} + {D_i}s})( {\hat{R}_i(s) - \hat{R}^{0}(s)} ),\nonumber
\end{align}
where
\begin{equation}\label{23}
\left\{ {\begin{aligned}
&{{I_i} = {I_0}},\\
&{{P_i} = {P_0} + \alpha i}, \\
&{{D_i} = {D_0} + \beta i}.
\end{aligned}} \right.i = 2,...,n.
\end{equation}
\textcolor{black}{The controller in \cite{18Khatir} is \emph{decentralized}, as it uses only on-board measurements (relative spacing and velocity). On the other hand, our proposed controller \eqref{LL:11} is \emph{cooperative}, as it involves communication of $a_{i-1}$ from the preceding vehicle.} Note that %it is straightforward to see that
the control gains in \eqref{23} increase linearly with the number of vehicles, which can bring large control effort for long platoons. On the other hand, %our work relaxes the linear increase in \eqref{23}, as
the proposed gains in \eqref{LL:18} can increase less than linearly. In addition, %the  performance in
\cite{18Khatir} only reaches \emph{eventual string stability} as in  Definition \ref{def1}, as compared to the stronger notion of SFSS in \textit{Theorem \ref{thm1}}.
\end{remark}
\subsection{\textcolor{black}{Robustness to time lags}}\label{sec:03c}
\textcolor{black}{In place of second-order vehicle dynamics let us now consider the third-order vehicle dynamics suggested in some platooning literature \cite{13Ploeg} \cite{Liu} \cite{a6} \cite{H}, with acceleration dynamics} ${{\dot a}_i} = \frac{1}{\tau }\left( -{{a_i} + {u_i}} \right)$: here, $\tau$ is a positive time constant representing the engine time constant, \textcolor{black}{i.e. it can approximate an actuator time lag}.

Consider the new error state ${{\bar e}_{i,3}} = {a_{i - 1}} - {a_i}$, which results in the third-order error dynamicS
\begin{equation}\label{LL:24}
\begin{aligned}
{{\dot e}_{i,1}(t)} &= {e_{i,2}(t)},\\
{{\dot e}_{i,2}(t)} &= {\bar e}_{i,3}(t),\\
{{\dot {\bar e}}_{i,3}(t)} &=  - \frac{1}{\tau }{\bar e_{i,3}(t)} - \frac{1}{\tau }{u_i(t)} + \frac{1}{\tau }{u_{i - 1}(t)}.
\end{aligned}
\end{equation}
 Consider the non-identical %feedforward and feedback
control law
\begin{equation}\label{LL:25}
{u_i}(t) = \alpha_i {e_{i,1}(t)} + \beta_i {e_{i,2}(t)} + \gamma_i {{\bar e}_{i,3}(t)} + \lambda_i {u_{i - 1}(t)}.
\end{equation}
As in Sect. \ref{sec:03a}, let us show that, also for the third-order dynamics \eqref{LL:24}, the identical control approach $\alpha_i=\alpha$, $\beta_i=\beta$, $\gamma_i=\gamma$, $\lambda_i=\lambda$ leads to degenerate string stability. In fact,  we obtain the  complementary sensitivity function $G_{i-1,i}(s)$ as follows
\begin{equation}\label{LL:26}
\frac{{{{\hat e}_{i,1}(s)}}}{{{{\hat u}_i}(s)}} = \frac{{1 - \lambda }}{{\tau \lambda {s^3} + \left( {\gamma  + \lambda } \right){s^2} + \beta s + \alpha }},
\end{equation}
\begin{align}\label{LL:27}
\left|G_{i-1,i}(s)\right|=& \left| \frac{{{{\hat e}_{i ,1}(s)}}}{{{{\hat e}_{i-1,1}(s)}}} \right|, \nonumber\\
= & \left|\frac{{\tau \lambda {s^3} + \left( {\gamma  + \lambda } \right){s^2} + \beta s + \alpha }}{{\tau {s^3} + \left( {\gamma  + 1} \right){s^2} + \beta s + \alpha }}\right|.
\end{align}
Accordingly, \eqref{LL:27} can be transformed into a polynomial of $\omega$ in order to meet condition \eqref{LL:09} as
%\begin{equation}\label{LL:28}
\begin{align}\label{LL:28}
&{\tau ^2}\left( {{\lambda ^2} - 1} \right){\omega ^4} + \left( {\lambda  - 1} \right)\left( {\lambda  + 1 + 2\gamma  - 2\tau \beta } \right){\omega ^2} \nonumber\\
&\quad- 2\alpha \left( {\lambda  - 1} \right)  \le 0, \ \forall \omega \nonumber\\
& {\lambda ^2} - 1 \le 0 \Rightarrow 0 < \lambda  \le 1, \nonumber\\
&2\alpha \left( {\lambda  - 1} \right) \ge 0 \Rightarrow \lambda  \ge 1.
\end{align}
%\end{equation}
The conditions in \eqref{LL:28} can be satisfied only for $\lambda = 1$, which makes $G_{i-1,i}(s)$ in \eqref{LL:27} identically equal to 1. This degenerate scenario is a major reason why most literature on third-order vehicle dynamics, in order to attain string stability with identical control, replaces the constant spacing policy \eqref{LL:06a} with the velocity-dependent spacing policy $e_{i,1} = R_i - R^*- hv_i$, being $h>0$ the time-headway constant \cite{13Ploeg,a6,a7}. This velocity-dependent spacing policy gives \textcolor{black}{string stability via identical control, since the complementary sensitivity function becomes $G_{i-1,i}(s) = 1/(1+hs)$ instead of 1.}

\textcolor{black}{Similar to what shown in Sect. \ref{sec:03b}, let us now attain string stability with constant-spacing policy \eqref{LL:06a} via the non-identical control \eqref{LL:25}. Using similar steps as in the proof of \textit{Theorem \ref{thm1}}, the complementary sensitivity function $G_{i-1,i}(s)$ for the third-order dynamics \eqref{LL:24} can be calculated as
\begin{align}\label{new1}
\left|G_{i\!-\!1,i}(s)\right|\!=& \left| \frac{{{{\hat e}_{i ,1}(s)}}}{{{{\hat e}_{i-1,1}(s)}}} \right|, \nonumber\\
 =& \left|\frac{{1 \!-\! {\lambda _i}}}{{1 \!-\! {\lambda _{i \!-\! 1}}}}\frac{{\tau {\lambda _{i \!-\! 1}}{s^3} \!+\! ({\gamma _{i \!-\! 1}} \!+\! {\lambda _{i \!-\! 1}}){s^2} \!+\! {\beta _{i \!-\! 1}}s \!+ \! {\alpha _{i \!-\! 1}}}}{{\tau {s^3} \!+\! ({\gamma _i} \!+\! 1){s^2} \!+\! {\beta _i}s \!+\! {\alpha _i}}} \right|
\end{align}
and the resulting polynomial of $\omega$ satisfying \eqref{LL:09} is
\begin{align}\label{new2}
&{\omega ^6}{\tau ^2}[{(1 \!-\! {\lambda _i})^2}{\lambda _{i \!-\! 1}^2} \!-\! {(1 \!-\! {\lambda _{i \!-\! 1}})^2}]\!+\! {\omega ^4}[{(1 \!-\! {\lambda _i})^2}( - 2\tau {\lambda _{i \!-\! 1}}{\beta _{i \!-\! 1}} \nonumber\\
& \!+\! {({\gamma _{i\! -\! 1}} \!+\! {\lambda _{i \!-\! 1}})^2}) \!-\! {(1 \!-\! {\lambda _{i \!-\! 1}})^2}(\! - 2\tau {\beta _i} \!+\! {({\gamma _i} \!+\! 1)^2})]\!+\!{\omega ^2}[{(1 \!-\! {\lambda _i})^2} \nonumber\\
&({\beta _{i \!-\! 1}^2} \!-\! 2{\alpha _{i \!-\! 1}}({\gamma _{i \!-\! 1}} \!+\! {\lambda _{i \!-\!1}})) \!-\! {(1 \!-\! {\lambda _{i \!-\! 1}})^2}({\beta _i^2} \!-\! 2{\alpha _i}({\gamma _i} \!+ \!1))] \nonumber\\
&+{(1 \!-\! {\lambda _i})^2}{\alpha _{i \!-\! 1}}^2 \!-\! {(1 \!-\! {\lambda _{i \!-\! 1}})^2}{\alpha _i}^2 \le 0, \ \forall \omega.
\end{align}
We obtain the following results.
\begin{corollary} \label{cor1}
[Robustness of longitudinal string stability to actuation lags]
Consider the third-order longitudinal dynamics \eqref{LL:24} with time lag $\tau>0$ and with the non-identical cooperative controller \eqref{LL:25}. Then, the tracking errors
\eqref{LL:06a} converge asymptotically when the control gains in \eqref{LL:25} are all positive. Moreover, the system is longitudinal SFSS if the control gains also satisfy the following conditions
\begin{equation}\label{new3}
\left\{ {\begin{aligned}
&{\lambda _i} \le \frac{1}{\lambda_{i-1}}, {\lambda _i}+\frac{1}{\lambda_{i-1}} \ge 2,\\
&{{\alpha _{i - 1}} \leq \frac{{1 \!-\! {\lambda _{i - 1}}}}{{1 \!-\! {\lambda _i}}}{\alpha _i}},\\
&{\beta _{i \!-\! 1}^2 \! \leq \! {{\left( \!{\frac{{1 \!-\! {\lambda _{i \!-\! 1}}}}{{1 \!-\! {\lambda _i}}}} \! \right)}^2}\!\left( { \!- 2{\alpha _i}({\gamma _i} \!+\! \lambda_{i\!-\!1}) \!+\! \beta _i^2} \right) \!+\! 2{\alpha _{i\! -\! 1}}({\gamma _{i \!-\! 1}}}\!+\! 1),\\
&\!\!-\!\! 2\tau {\lambda _{i \!-\! 1}}{\beta _{i \!-\! 1}} \!\!+\! {({\gamma _{i \!- \!1}} \!\!+\! {\lambda _{i \!-\! 1}})^2} \!\le \!{{\left( \!{\frac{{1 \!-\! {\lambda _{i \!-\! 1}}}}{{1 \!-\! {\lambda _i}}}} \! \right)}^2}\!( \!-\! 2\tau {\beta _i} \!+\! {({\gamma _i} \!+\! 1)^2}).
\end{aligned}} \right.
\end{equation}
\end{corollary}
}

\subsection{Robustness to V2V communication delay}
When the preceding vehicle communicates wirelessly its own acceleration $a_{i-1}$, a communication delay may arise and it becomes crucial to study how to make string stability robust against this delay \cite{new3}. To this purpose, \eqref{LL:11} is modified as
\begin{equation}\label{LL:29}
{a_i}\left( t \right) = {\alpha _i}{e_{i,1}}\left( t \right) + {\beta _i}{e_{i,2}}\left( t \right) + {\gamma _i}{a_{i - 1}}\left( {t - {t_d}} \right)
\end{equation}
being $t_d$  a constant communication delay. To make the analysis of SFSS analytically tractable, we adopt a first-order Pade approximation technique as in \cite{new4} %, which means
\begin{equation}\label{LL:30}
%{{\hat a}_{i - 1}}\left( {t - {t_d}} \right) =
{e^{ - {t_d}s}}{{\hat a}_{i - 1}}\left( s \right) \approx \frac{{1 - \frac{{{t_d}s}}{2}}}{{1 + \frac{{{t_d}s}}{2}}}{{\hat a}_{i - 1}}\left( s \right)
\end{equation}
where ${e^{ - {t_d}s}}{{\hat a}_{i - 1}}\left( s \right)$ is the Laplace transform of ${a_{i - 1}}\left( {t - {t_d}} \right)$. By introducing an auxiliary variable ${{\bar x}_{i - 1}}$, it is possible to write \eqref{LL:30} as a first-order dynamical system
\begin{equation}\label{LL:31}
 {{\dot {\bar x}}_{i - 1}}\left( t \right) = -\frac{2}{t_d}{{\bar x}_{i - 1}}\left( t \right)  +\frac{4}{{t_d}}{a_{i - 1}}\left( t \right).
\end{equation}
Accordingly, the delayed feedforward term ${a_{i - 1}}\left( {t - {t_d}} \right)$ can be approximated as
\begin{equation}\label{LL:32}
\begin{aligned}
{a_{i - 1}}\left( {t - {t_d}} \right) = {{\bar x}_{i - 1}}\left( t \right) - {a_{i - 1}}\left( t \right).
\end{aligned}
\end{equation}
%Based on the auxiliary error variable ${{\bar x}_{i - 1}}$,
We obtain the new error dynamics as
\begin{equation}\label{LL:33}
\left[ { \setlength{\arraycolsep}{1.5pt}\begin{array}{*{20}{c}}
{{{\dot e}_{i,1}}\left( t \right)}\\
{{{\dot e}_{i,2}}\left( t \right)}\\
{{{\dot {\bar x}}_{i - 1}}\!\left( t \right)}
\end{array}} \right] \!\!=\!\! \left[  \setlength{\arraycolsep}{1.2pt}{\begin{array}{*{20}{c}}
0&1&0\\
{ - {\alpha _i}}&{ - {\beta _i}}& -\gamma_i\\
0&0& - \frac{2}{t_d}
\end{array}} \right]\!\!\!\left[  \setlength{\arraycolsep}{1.5pt}{\begin{array}{*{20}{c}}
{{e_{i,1}}\left( t \right)}\\
{{e_{i,2}}\left( t \right)}\\
{{{\bar x}_{i - 1}}\!\left( t \right)}
\end{array}} \right] \!+ \!\left[  \setlength{\arraycolsep}{1pt}{\begin{array}{*{20}{c}}
0\\
{{1\!+\!\gamma _i}}\\
{ \frac{4}{t_d}}
\end{array}} \right]\!{a_{i - 1}}\!\left( t \right)
\end{equation}
It is clear that dynamics  \eqref{LL:33}  are asymptotically stable when $\alpha_i,\beta_i>0$ and $a_{i-1}(t)\!=\!0$. Then, we explore the influence of delay on string stability with the following result.
\begin{theorem}\label{thm2}
[Robustness of longitudinal string stability to time-delay] Consider the longitudinal dynamics \eqref{LL:33} and the non-identical cooperative controller \eqref{LL:29} with  constant time delay $t_d$. The system is longitudinal SFSS when the control gains satisfy
\begin{equation}\label{LL:34}
\left\{\begin{aligned}
 & 0<{\gamma _i} \le \frac{1}{\gamma_{i-1}}, {\gamma _i} + \frac{1}{{{\gamma _{i - 1}}}} \ge 2,\\
 & {\alpha _{i - 1}} \le \frac{{1 - {\gamma _{i - 1}}}}{{1 - {\gamma _i}}}{\alpha _i},\\
% & \! \frac{1}{q_i^{2}} \!\left( {\frac{{\beta _{i \!-\! 1}^2}}{{\gamma _{i \!-\! 1}^2}} \!-\! 2\frac{{{\alpha _{i \!-\! 1}}}}{{{\gamma _{i \!-\! 1}}}}} \right) \!<\beta_i^2 \!-\! 2\alpha_i <\frac{1}{p_i^{2}}\!\left( {2\frac{{{\alpha _{i \!-\! 1}}}}{{{\gamma _{i \!-\! 1}}}} \!+ \! \frac{{\beta _{i \!-\! 1}^2}}{{\gamma _{i \!-\! 1}^2}}} \right)\!.\\
&t_d^2\!\left(\! {2\frac{{{\alpha _{i \!-\! 1}}}}{{{\gamma _{i \!-\! 1}}}} \!+\! \frac{{\beta _{i \!-\! 1}^2}}{{\gamma _{i \!-\! 1}^2}} \!+\! p_i^2\!\left( {2{\alpha _i} \!-\! \beta _i^2} \right)} \!\right) \!+\! 8{t_d}\frac{{{\beta _{i \!-\! 1}}}}{{{\gamma _{i \!-\! 1}}}} \!+\! 4\!\left( {1 \!-\! p_i^2} \right) \!\le \!0,\\
&t_d^2\!\left(\! {2\frac{{{\alpha _{i \!-\! 1}}}}{{{\gamma _{i \!-\! 1}}}} \!+\! \frac{{\beta _{i \!-\! 1}^2}}{{\gamma _{i \!-\! 1}^2}} \!+\! q_i^2\!\left( {2{\alpha _i} \!-\! \beta _i^2} \right)}\! \right) \!+\! 8{t_d}\frac{{{\beta _{i \!-\! 1}}}}{{{\gamma _{i \!-\! 1}}}} \!+\! 4\!\left( {1 \!-\! q_i^2} \right)\! \le\! 0.
\end{aligned}
 \right.
\end{equation}
with  $q_i=\frac{{\frac{1}{{{\gamma _{i - 1}}}} - 1}}{{1 - {\gamma _i}}}$, $p_i=\frac{{\frac{1}{{{\gamma _{i - 1}}}} + 1}}{{1 + {\gamma _i}}}$.
\end{theorem}
\begin{IEEEproof}
According to the error dynamics \eqref{LL:33}, we can obtain the transfer function from $a_{i-1}$ to $e_{i,1}$ and $e_{i-1,1}$ as
\begin{equation}\label{LL:35}
\begin{aligned}
\frac{{{{\hat e}_{i,1}}\left( s \right)}}{{{{\hat a}_{i - 1}}\left( s \right)}} &= \frac{{1 - {\gamma _i}\frac{{2 - {t_d}s}}{{2 + {t_d}s}}}}{{{s^2} + {\alpha _i} + {\beta _i}s}},\\
\frac{{{{\hat e}_{i - 1,1}}\left( s \right)}}{{{{\hat a}_{i - 1}}\left( s \right)}} &= \frac{{\frac{{2 + {t_d}s}}{{2 - {t_d}s}}\frac{1}{{{\gamma _{i - 1}}}} - 1}}{{{s^2} + \left( {{\alpha _{i - 1}} + {\beta _{i - 1}}s} \right)\frac{{2 + {t_d}s}}{{2 - {t_d}s}}\frac{1}{{{\gamma _{i - 1}}}}}}.
\end{aligned}
\end{equation}
Therefore, the complementary sensitivity function in \eqref{LL:09} is
\begin{equation}\label{LL:36}
\begin{aligned}
\left|G_{i-1,i}(s)\right|&= \! \left| {\frac{{{{\hat e}_{i,1}}\left( s \right)}}{{{{\hat e}_{i - 1,1}}\left( s \right)}}} \right| \\
&=\left| {\frac{{1 \!-\! {\gamma _i}\frac{{2 - {t_d}s}}{{2 + {t_d}s}}}}{{\frac{{2 + {t_d}s}}{{2 - {t_d}s}}\frac{1}{{{\gamma _{i \!-\! 1}}}} \!-\! 1}}\frac{{{s^2} \!+\! \left( {{\alpha _{i \!-\! 1}} \!+\! {\beta _{i \!-\! 1}}s} \right)\frac{{2 + {t_d}s}}{{2 - {t_d}s}}\frac{1}{{{\gamma _{i \!-\! 1}}}}}}{{{s^2} \!+\! {\alpha _i} \!+\! {\beta _i}s}}} \right|,\\
&=\left| {\frac{{{b_0}{s^4} + {b_1}{s^3} + {b_2}{s^2} + {b_3}s + {b_4}}}{{{c_0}{s^4} + {c_1}{s^3} + {c_2}{s^2} + {c_3}s + {c_4}}}} \right|
\end{aligned}
\end{equation}
with the coefficients being
$b_0= - t_d^2\left( {1 + {\gamma _i}} \right)$,\\
$b_1={ - 2{t_d}\left( {1 - {\gamma _i}} \right) + {t_d}\left( {1 + {\gamma _i}} \right)\left( {2 + {t_d}\frac{{{\beta _{i - 1}}}}{{{\gamma _{i - 1}}}}} \right)}$,\\
$b_2={2\left( {1 - {\gamma _i}} \right)\left( {2 + {t_d}\frac{{{\beta _{i - 1}}}}{{{\gamma _{i - 1}}}}} \right) + \frac{{\left( {{t_d}{\alpha _{i - 1}} + 2{\beta _{i - 1}}} \right)}}{{{\gamma _{i - 1}}}}{t_d}\left( {1 + {\gamma _i}} \right)}$,\\
$b_3={2\left( {1 - {\gamma _i}} \right)\frac{{\left( {{t_d}{\alpha _{i - 1}} + 2{\beta _{i - 1}}} \right)}}{{{\gamma _{i - 1}}}} + 2\frac{{{\alpha _{i - 1}}}}{{{\gamma _{i - 1}}}}{t_d}\left( {1 + {\gamma _i}} \right)}$,\\
$b_4=4\frac{{{\alpha _{i - 1}}}}{{{\gamma _{i - 1}}}}\left( {1 - {\gamma _i}} \right)$; and
$c_0= - t_d^2\left( {\frac{1}{{{\gamma _{i - 1}}}} + 1} \right)$,\\
$c_1={ - 2{t_d}\left( {\frac{1}{{{\gamma _{i - 1}}}} - 1} \right) + {t_d}\left( {\frac{1}{{{\gamma _{i - 1}}}} + 1} \right)\left( {2 - {t_d}{\beta _i}} \right)}$,\\
$c_2={2\left( {\frac{1}{{{\gamma _{i - 1}}}} - 1} \right)\left( {2 - {t_d}{\beta _i}} \right) + {t_d}\left( {\frac{1}{{{\gamma _{i - 1}}}} + 1} \right)\left( {2{\beta _i} - {\alpha _i}{t_d}} \right)}$,\\
$c_3={2\left( {\frac{1}{{{\gamma _{i - 1}}}} - 1} \right)\left( {2{\beta _i} - {\alpha _i}{t_d}} \right) + 2{\alpha _i}{t_d}\left( {\frac{1}{{{\gamma _{i - 1}}}} + 1} \right)}$,\\
$c_4=4{\alpha _i}\left( {\frac{1}{{{\gamma _{i - 1}}}} - 1} \right)$.\\
The $H_\infty$ norm of \eqref{LL:36} results in an eighth-order polynomial in the variable $\omega$. Sufficient conditions can be obtained by making the coefficients of the polynomial negative, yielding
\begin{equation}\label{LL:37}
\begin{aligned}
&t_d^4{\left( {1 + {\gamma _i}} \right)^2} \le t_d^4{\left( {\frac{1}{{{\gamma _{i - 1}}}} + 1} \right)^2},\\
&16{\left( {1 - {\gamma _i}} \right)^2}\frac{{\alpha _{i - 1}^2}}{{\gamma _{i - 1}^2}} \le 16{\left( {\frac{1}{{{\gamma _{i - 1}}}} - 1} \right)^2}\alpha _i^2,\\
&{{{\left( {1 + {\gamma _i}} \right)}^2}\frac{{\alpha _{i - 1}^2}}{{\gamma _{i - 1}^2}} \le {{\left( {\frac{1}{{{\gamma _{i - 1}}}} + 1} \right)}^2}\alpha _i^2},\\
&{\left( {1 - {\gamma _i}} \right)^2} \le {\left( {\frac{1}{{{\gamma _{i - 1}}}} - 1} \right)^2},\\
%\end{aligned}
%\end{equation}
%\begin{equation}\label{LL:37}
%\begin{aligned}
 &t_d^2\left( {{{\left( {1 \!+\! {\gamma _i}} \right)}^2}\left( {2\frac{{{\alpha _{i - 1}}}}{{{\gamma _{i - 1}}}} \!+\! \frac{{\beta _{i - 1}^2}}{{\gamma _{i - 1}^2}}} \right) \!+\! {{\left( {\frac{1}{{{\gamma _{i - 1}}}} \!+\! 1} \right)}^2}\left( {2{\alpha _i} \!-\! \beta _i^2} \right)} \right) \\
&+ 8{t_d}{\left( {1 \!+\! {\gamma _i}} \right)^2}\frac{{{\beta _{i - 1}}}}{{{\gamma _{i - 1}}}} \!+\! 4\left( {{{\left( {1 \!+\! {\gamma _i}} \right)}^2} \!-\! {{\left( {\frac{1}{{{\gamma _{i - 1}}}} \!+\! 1} \right)}^2}} \right) \le 0,\\
&t_d^2\left( {{{\left( {1 \!-\! {\gamma _i}} \right)}^2}\left( {2\frac{{{\alpha _{i - 1}}}}{{{\gamma _{i - 1}}}} \!+\! \frac{{\beta _{i - 1}^2}}{{\gamma _{i - 1}^2}}} \right) \!+\! {{\left( {\frac{1}{{{\gamma _{i - 1}}}} \!-\! 1} \right)}^2}\left( {2{\alpha _i} \!-\! \beta _i^2} \right)} \right) \\
&+ 8{t_d}{\left( {1 \!-\! {\gamma _i}} \right)^2}\frac{{{\beta _{i - 1}}}}{{{\gamma _{i - 1}}}} \!+\! 4\left({\left( {1 - {\gamma _i}} \right)^2} \!-\! {\left( {\frac{1}{{{\gamma _{i - 1}}}}\! -\! 1} \right)^2} \right)\le 0.
\end{aligned}
\end{equation}
From \eqref{LL:37}, the conditions \eqref{LL:34} can be obtained after suitable rearrangements, which completes the proof.
\end{IEEEproof}
\begin{remark}
[Upper bound of time-delay] From \eqref{LL:37}, one can see that a large $\beta_i$  provides string stability for large delays. On the other hand, a large $\beta_i$ may increase possible noise entering the error derivative. Thus, it is practical to design the control gains in \eqref{LL:34} based on a maximum delay, e.g. considering V2V delays in the typical range $t_d\in [0.02, 0.1]s$, \cite{new4}. % and that the Pade approximation is reasonable in this range.
%With the given range of $t_d$, one can make the the longitudinal controller \eqref{LL:29} robust via the control gains requirements \eqref{LL:34}.
\end{remark}

\section{Lateral control}\label{sec:04}
We show that the non-identical approach is not necessary for lateral string stability, i.e. identical gains are enough to guarantee SFSS in the lateral direction. Thus, we adopt %the identical %feedforward and feedback
%protocol
\begin{equation}\label{LL:39}
{\omega_i}(t) = k_3e_{i,3}(t)+k_4e_{i,4}(t) + {\mu }{\omega_{i - 1}(t)},
\end{equation}
where $k_3, k_4$ and $\mu$ are control gains to be designed. %In the following, we will show that the non-identical approach is not necessary for lateral string stability. % of lateral dynamics.
\begin{remark}
[Cooperative lateral control] Different from the decentralized approach of \cite{18Khatir}, %cooperative control allows the predecessor to communicate with its follower by wireless communication. Therefore,
\eqref{LL:39} considers cooperative control also in the lateral direction.
\end{remark}
\subsection{Proposed cooperative lateral control}
After substituting the control law \eqref{LL:39} into \eqref{LL:03b}, the resulting error dynamics are
\begin{equation}\label{LL:40}
\begin{bmatrix}{\dot e}_{i,3}(t) \\ {\dot e}_{i,4}(t) \end{bmatrix}=\begin{bmatrix}-k_3 & \frac{v^*}{R^*}-k_4\\-k_3& -k_4 \end{bmatrix}\begin{bmatrix} e_{i,3}(t) \\ e_{i,4}(t)\end{bmatrix}+\begin{bmatrix} -\mu \\ 1-\mu\end{bmatrix}\omega_{i-1}(t).
\end{equation}
Consider the Laplace transform of \eqref{LL:40} resulting in
\begin{equation}\label{LL:41}
\begin{aligned}
\frac{{{{\hat e}_{i,3}(s)}}}{{{{\hat \omega }_{i - 1}(s)}}} = \frac{{ - \mu s + \frac{{{v^*}}}{{{R^*}}}(1 - \mu ) - {k_4}}}{{{s^2} + \left( {{k_3} + {k_4}} \right)s + \frac{{{v^*}}}{{{R^*}}}{k_3}}},\\
\frac{{{{\hat e}_{i,4}(s)}}}{{{{\hat \omega }_{i - 1}(s)}}} = \frac{{(1 - \mu )s + {k_3}}}{{{s^2} + \left( {{k_3} + {k_4}} \right)s + \frac{{{v^*}}}{{{R^*}}}{k_3}}}.
\end{aligned}
\end{equation}
According to the Routh-Hurwitz criterion, the second-order error dynamics \eqref{LL:40} are asymptotically stable if and only if
\begin{equation}\label{LL:42}
{k_3} > 0, \ {k_4} > 0.
\end{equation}
In addition to asymptotic stability, lateral string stability is also to be ensured, leading to the following result.

\begin{theorem}\label{thm3} [Lateral stability and string stability] Consider the lateral dynamics in \eqref{LL:03b} with the error states in \eqref{LL:07a}. The lateral controller \eqref{LL:39} with identical control gains ensures asymptotic stability when the control gains are all positive. In addition, lateral SFSS can be guaranteed if the following conditions are satisfied
\begin{equation} \label{LL:43}
\left\{ \begin{array}{l}
0< \mu  \le 1,\\
{k_3} + 2{k_4} \geq 2\frac{{{v^*}}}{{{R^*}}}\left( {1 - \mu } \right).
\end{array} \right.
\end{equation}
\end{theorem}

\begin{IEEEproof}
To satisfy \eqref{LL:42}, calculate % being satisfied, straightforward manipulations give %, there is
\begin{equation}\label{LL:44}
\begin{aligned}
\frac{{{{\hat e}_{i,3}(s)}}}{{{{\hat \omega }_i}(s)}} &= \frac{{ - \mu s - {k_4} + \frac{{{v^*}}}{{{R^*}}}\left( {1 - \mu } \right)}}{{\mu {s^2} + {k_4}s + \frac{{{v^*}}}{{{R^*}}}{k_3}}},\\
\frac{{{{\hat e}_{i,4}(s)}}}{{{{\hat \omega }_i}(s)}} &= \frac{{\left( {1 - \mu } \right)s + {k_3}}}{{\mu {s^2} + {k_4}s + \frac{{{v^*}}}{{{R^*}}}{k_3}}}.
\end{aligned}
\end{equation}
Referring to the complementary sensitivity function notation in \eqref{LL:09}, combine \eqref{LL:41} and \eqref{LL:44}, for $i=2,\dots,n$,
\begin{align}\label{LL:45}
\left|G_{i-1,i}(s)\right|=& \left|  {\frac{{{{\hat e}_{i ,3}(s)}}}{{{{\hat e}_{i-1,3}(s)}}}} \right| =  \left| {\frac{{{{\hat e}_{i ,4}}(s)}}{{{{\hat e}_{i-1,4}(s)}}}}  \right| \nonumber \\ =&\left| \frac{{\mu {s^2} + {k_4}s + \frac{{{v^*}}}{{{R^*}}}{k_3}}}{{{s^2} + \left( {{k_3} + {k_4}} \right)s + \frac{{{v^*}}}{{{R^*}}}{k_3}}}\right|.
\end{align}
The condition \eqref{LL:09} for SFSS can be written as
\begin{equation}\label{LL:46}
\left| {\frac{{ - \mu {\omega ^2} + \frac{{{v^*}}}{{{R^*}}}{k_3} + {k_4}\omega j}}{{ - {\omega ^2} + \frac{{{v^*}}}{{{R^*}}}{k_3} + \left( {{k_3} + {k_4}} \right)\omega j}}} \right| \leq 1,\forall \omega.
\end{equation}
which leads to a polynomial in the variable $\omega$ of the form
\begin{equation}\label{LL:47}
\left( { {\mu ^2}-1} \right){\omega ^2} +  {2\frac{{{v^*}}}{{{R^*}}}{k_3}\left( { 1-\mu } \right) - k_3^2 - 2{k_3}{k_4}} \leq 0,\forall \omega.
\end{equation}
Using the determinant conditions, we obtain the conditions \eqref{LL:43} on the control gains. This concludes the proof.
\end{IEEEproof}

\begin{remark}
[Identical lateral control] \textit{Theorem \ref{thm3}} achieves SFSS in the lateral direction, a stronger result than the eventual string stability in \cite{18Khatir}. Different from \textit{Theorem \ref{thm1}} for longitudinal control, \textit{Theorem \ref{thm3}} shows that identical control gains are sufficient for guaranteeing lateral string stability. This is due to the different form of the lateral dynamics as compared to the longitudinal dynamics, leading to a different  $G_{i-1,1}(s)$ (compare \eqref{LL:45} with \eqref{LL:20}). As a matter of fact, \cite{18Khatir} already noticed that the dynamics of $e_{i,3}$ and $e_{i,4}$ %in \eqref{LL:03b}
makes it complex to study lateral string stability using non-identical gains.
\end{remark}

\subsection{Robustness to V2V communication delay}
Similar to the longitudinal case, consider a communication delay $t_d$ also in the lateral case, via the term $\omega_{i-1}$, i.e.
\begin{equation}\label{LL:48}
{\omega _i}\left( t \right) = {k_3}{e_{i,3}}\left( t \right) + {k_4}{e_{i,4}}\left( t \right) + \mu {\omega _{i - 1}}\left( {t - {t_d}} \right).
\end{equation}
As in \eqref{LL:31}, introduce an auxiliary variable $\bar{y}_{i-1}$ with dynamics
\begin{equation}\label{LL:49}
 {{\dot {\bar y}}_{i - 1}}\left( t \right) = -\frac{2}{t_d}{{\bar y}_{i - 1}}\left( t \right)  +\frac{4}{{t_d}}{\omega_{i - 1}}\left( t \right).
\end{equation}
giving the transformed lateral error dynamics
\begin{equation}\label{LL:50}
\left[ { \setlength{\arraycolsep}{1.5pt}\begin{array}{*{20}{c}}
{{{\dot e}_{i,3}}\left( t \right)}\\
{{{\dot e}_{i,4}}\left( t \right)}\\
{{{\dot {\bar y}}_{i - 1}}\!\left( t \right)}
\end{array}} \!\right] \!\!=\!\! \left[  \setlength{\arraycolsep}{1.2pt}{\begin{array}{*{20}{c}}
-k_3&\ \frac{v^*}{R^*}\!-\!k_4&-\mu\\
 -k_3&-k_4&-\mu\\
0&0& -{\frac{{2 }}{{{t_d}}}}
\end{array}} \right]\!\!\!\left[  \setlength{\arraycolsep}{1pt}{\begin{array}{*{20}{c}}
{{e_{i,3}}\left( t \right)}\\
{{e_{i,4}}\left( t \right)}\\
{{{\bar y}_{i - 1}}\!\left( t \right)}
\end{array}}\! \right] \!+ \!\left[  \setlength{\arraycolsep}{1pt}{\begin{array}{*{20}{c}}
\mu\\
{{1\!+\!\mu}}\\
{ \frac{4}{t_d}}
\end{array}} \right]\!{\omega_{i - 1}}\!\left( t \right)
\end{equation}
With $\omega_{i-1}(t)=0$, stability of lateral dynamics \eqref{LL:50} holds when $k_3, k_4>0$. Robustness of lateral string stability is derived via the following result.
\begin{theorem}\label{thm4}
[Robustness of lateral string stability to time-delay] Consider the lateral error dynamics in \eqref{LL:50} and the cooperative  control law \eqref{LL:48} with constant time-delay $t_d$. Lateral SFSS is ensured when the control gains satisfy
\begin{equation}\label{LL:51}
\left\{\begin{aligned}
 &0 < \mu  \le 1,\\
 &{k_3} + 2{k_4} \ge 2\frac{{{v^*}}}{{{R^*}}}\left( {1 - \mu } \right),\\
 &t_d^2{k_3}\frac{{{v^*}}}{{{R^*}}}\left( {\mu  \!+\! 1} \right)\! +\! {t_d}{k_4}\left( {4\mu  \!-\! {k_3}} \right) \!+\! 2\left( {{\mu ^2} \!-\! 1} \right) \!\le 0.\\
\end{aligned}\right.
\end{equation}
\end{theorem}

\begin{figure}[t]
    \centering
    \setlength{\abovecaptionskip}{-0.1cm}
    \subfigure[Decentralized method in \cite{18Khatir}.]{%
    \label{figure3a}%
    \includegraphics[width=3.2in]{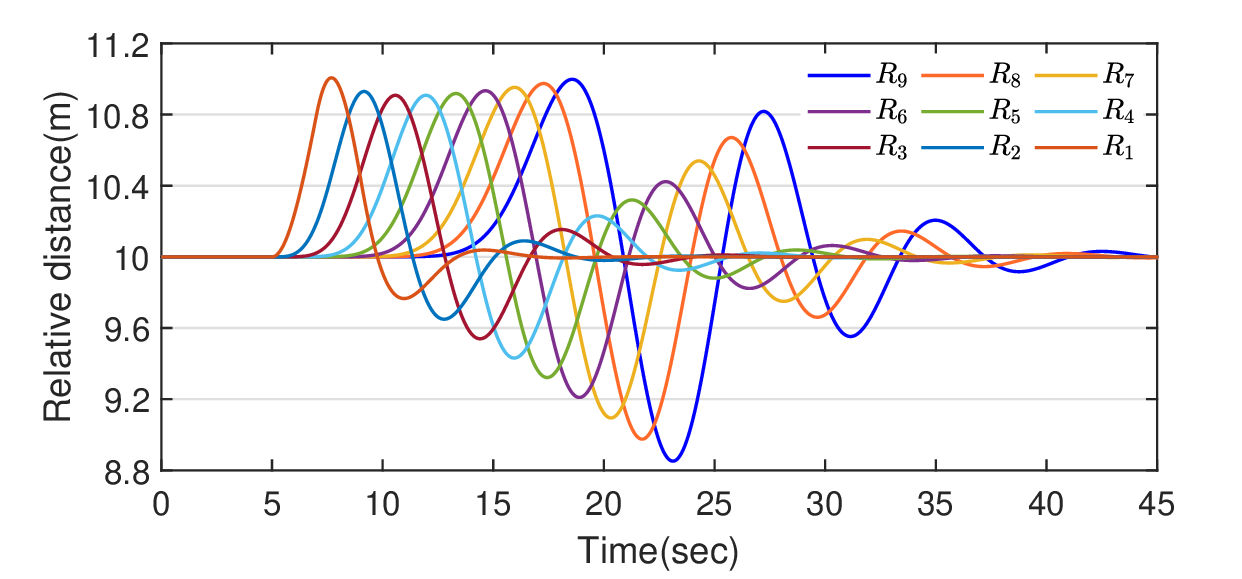}}
    %\hspace{10pt}%
    \subfigure[Proposed cooperative method.]{%
    \label{figure3b}%
    \includegraphics[width=3.2in]{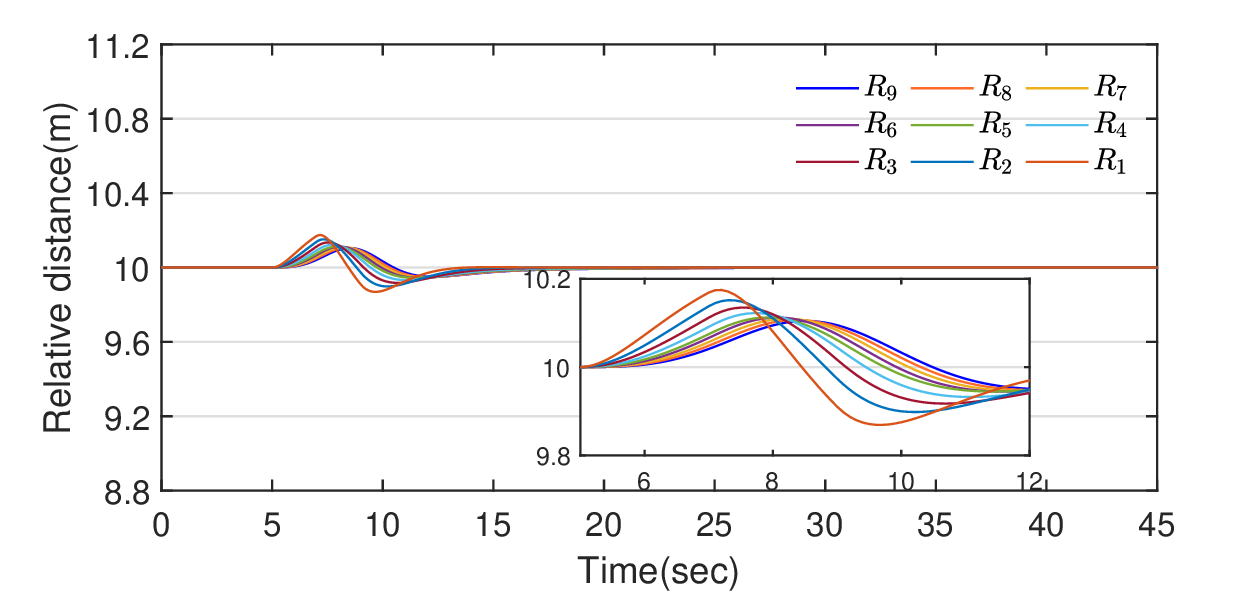}}
    \caption{\textcolor{black}{Longitudinal control: comparison of relative distances for the decentralized method in \cite{18Khatir} and the proposed cooperative method. The proposed method dramatically improves the inter-vehicle distance.} }%
    \label{distance}%
\end{figure}

\begin{figure}[t]
    \centering
    \setlength{\abovecaptionskip}{-0.1cm}
    \subfigure[Decentralized method in \cite{18Khatir}.]{%
    \label{figure4a}%
    \includegraphics[width=3.2in]{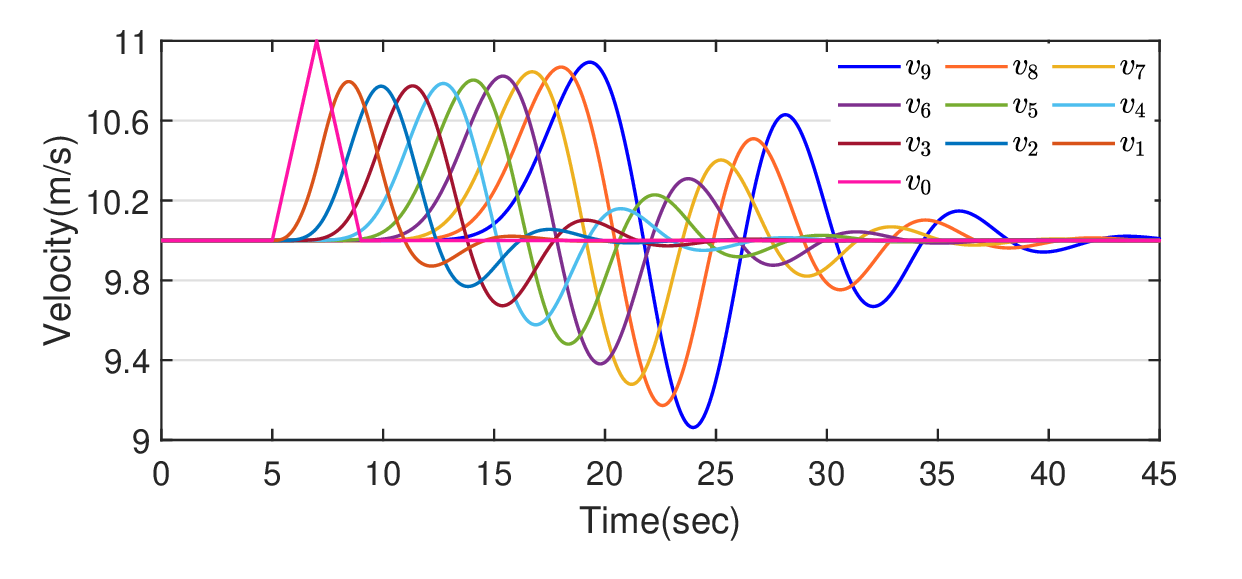}}
   % \hspace{10pt}%
    \subfigure[Proposed cooperative method.]{%
    \label{figure4b}%
    \includegraphics[width=3.2in]{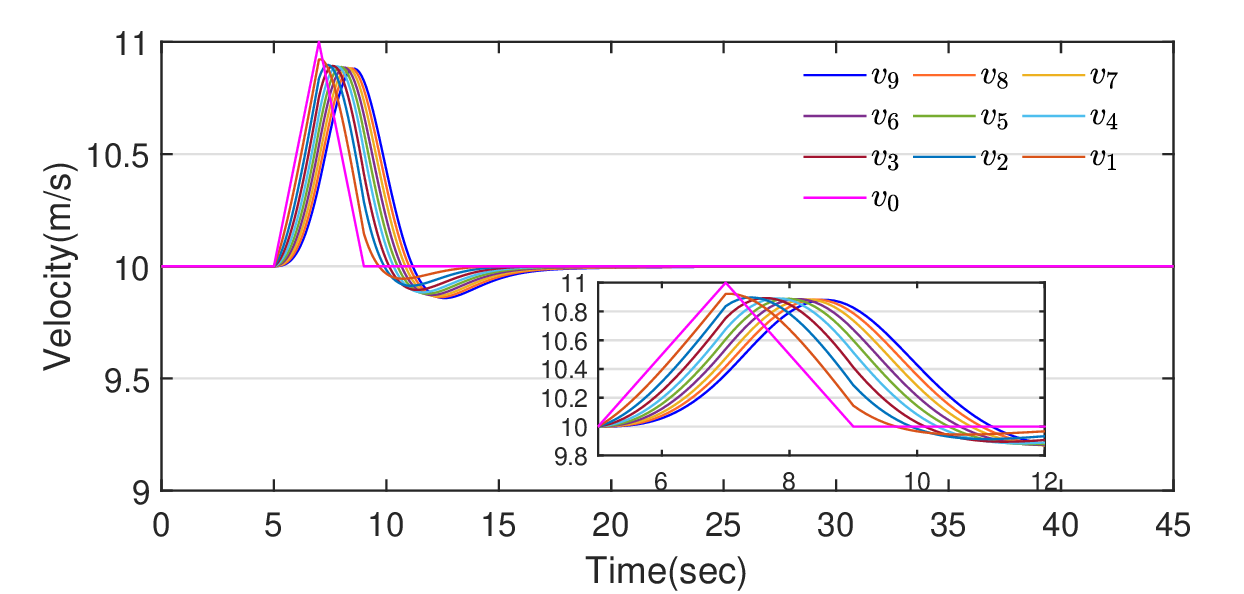}}
    \caption{\textcolor{black}{Longitudinal control: comparison of velocities for the decentralized method in \cite{18Khatir} and the proposed cooperative method. It can be seen that the proposed method dramatically improves the tracking of the leader's velocity. }}%
    \label{veelocity}%
\end{figure}

\begin{figure}[t]
    \centering
    \setlength{\abovecaptionskip}{-0.1cm}
    \subfigure[Distance error $e_{i,1}$ without communication delay.]{%
    \label{figure6a}%
    \includegraphics[width=3.2in]{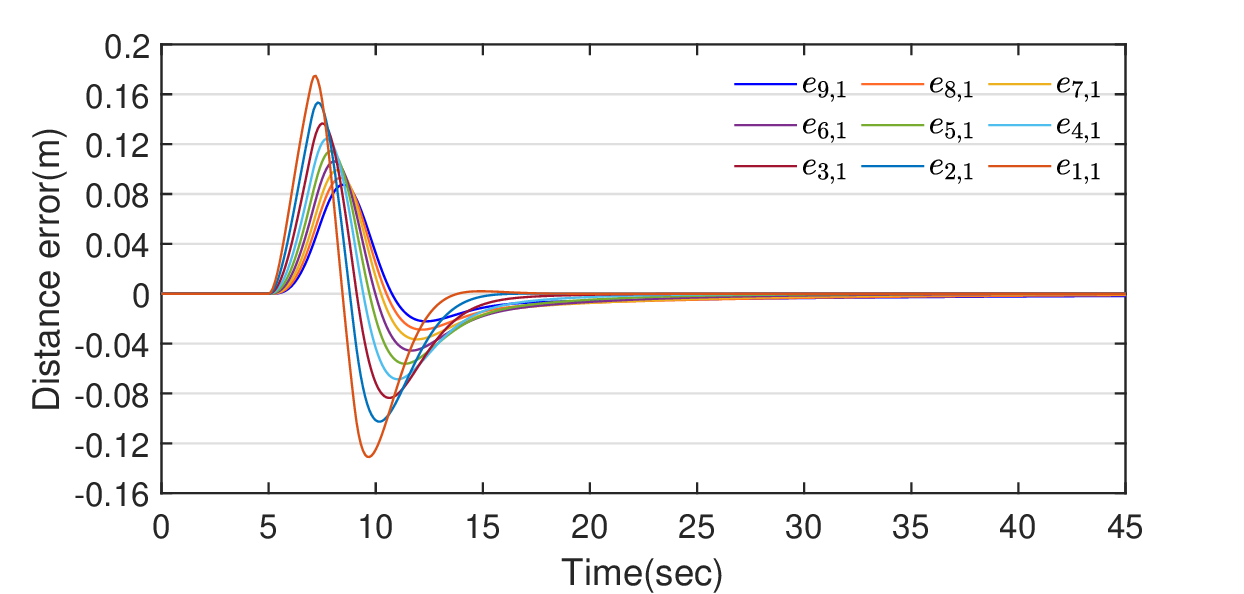}}
    \hspace{10pt}%
    \subfigure[Distance error $e_{i,1}$ with delay $t_d=0.1s$.]{%
    \label{figure9a}%
    \includegraphics[width=3.2in]{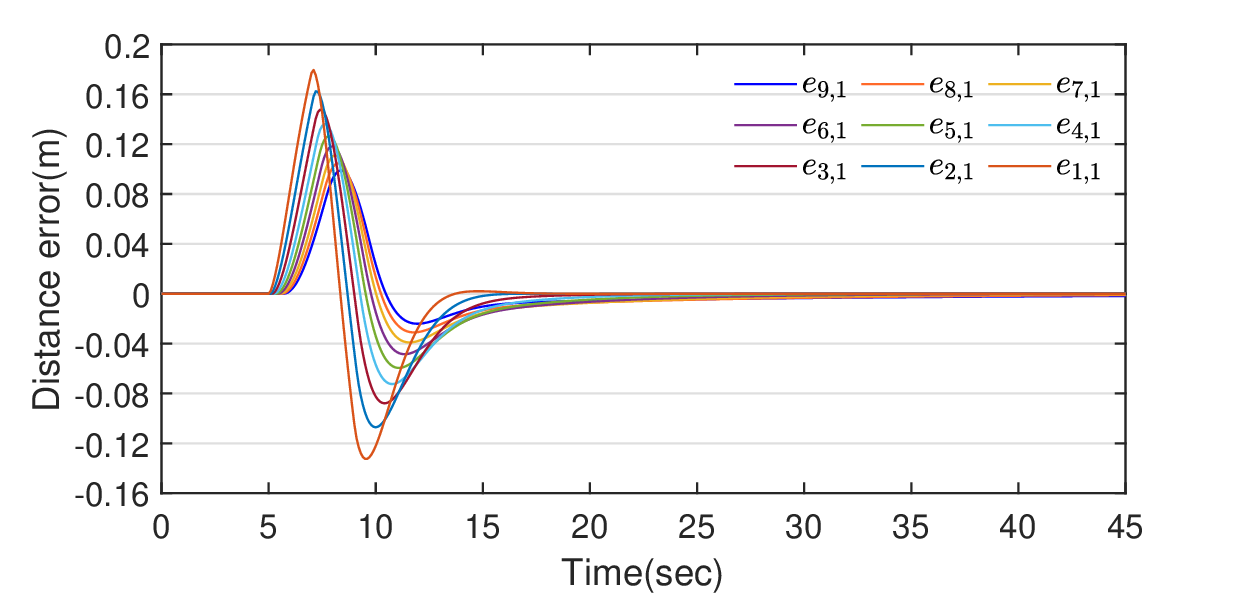}}
    \subfigure[Velocity error $e_{i,2}$ without communication delay .]{%
    \label{figure6b}%
    \includegraphics[width=3.2in]{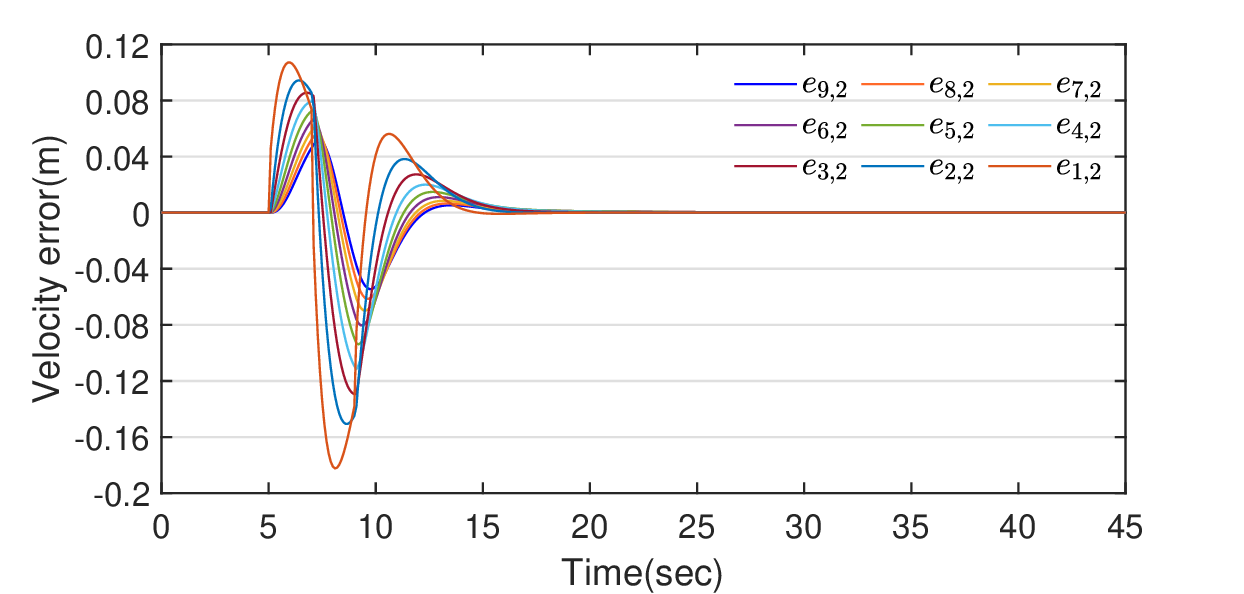}}
    \hspace{10pt}%
    \subfigure[Velocity error $e_{i,2}$ with delay $t_d=0.1s$.]{%
    \label{figure9b}%
    \includegraphics[width=3.2in]{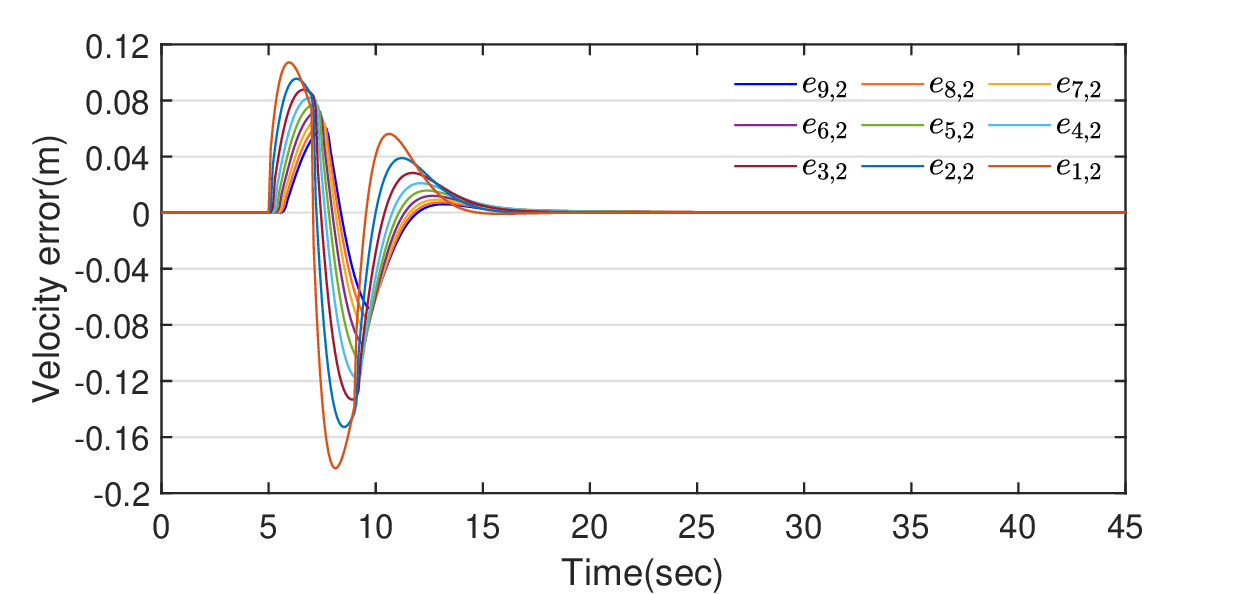}}
    \caption{\textcolor{black}{Longitudinal control: comparison of longitudinal errors  without communication delay and with delay $t_d=0.1$s for the proposed cooperative method. It can be seen that the proposed method is robust to such delay as there is no noticeable effect on the error signals.} }%
    \label{e1e2}%
\end{figure}

%\begin{figure}[tp]
%    \centering
%    \setlength{\abovecaptionskip}{-0.1cm}
%    \subfigure[Distance error $e_{i,1}$ with time-delay $t_d=0.1s$.]{%
%    \label{figure9a}%
%    \includegraphics[width=2.6in]{delaye1.eps}}
%    \hspace{10pt}%
%    \subfigure[Velocity error $e_{i,2}$ with time-delay $t_d=0.1s$.]{%
%    \label{figure9b}%
%    \includegraphics[width=2.6in]{delaye2.eps}}
%    \caption{Longitudinal errors for the vehicle platoons with time-delay.}%
%    \label{delaye1e2}%
%\end{figure}
\begin{figure}[t]
    \centering
    \setlength{\abovecaptionskip}{-0.1cm}
    \subfigure[No communication delay.]{%
    \label{figure5a}%
    \includegraphics[width=3.2in]{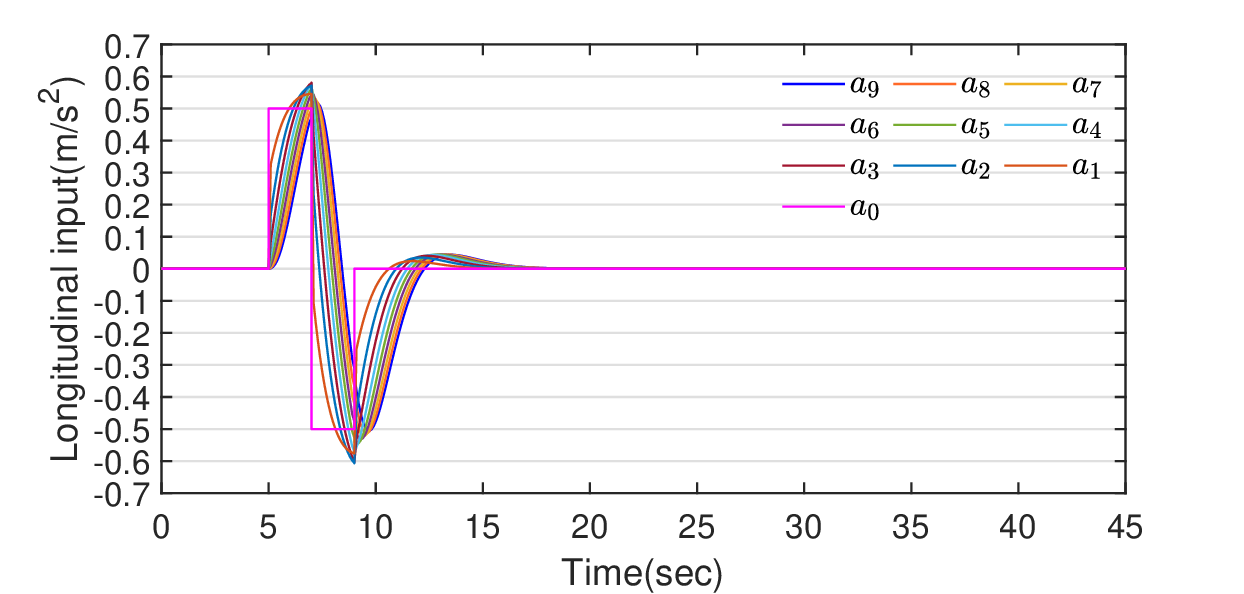}}
   % \hspace{10pt}%
    \subfigure[Communication delay $t_d=0.1s$.]{%
    \label{figure5b}%
    \includegraphics[width=3.2in]{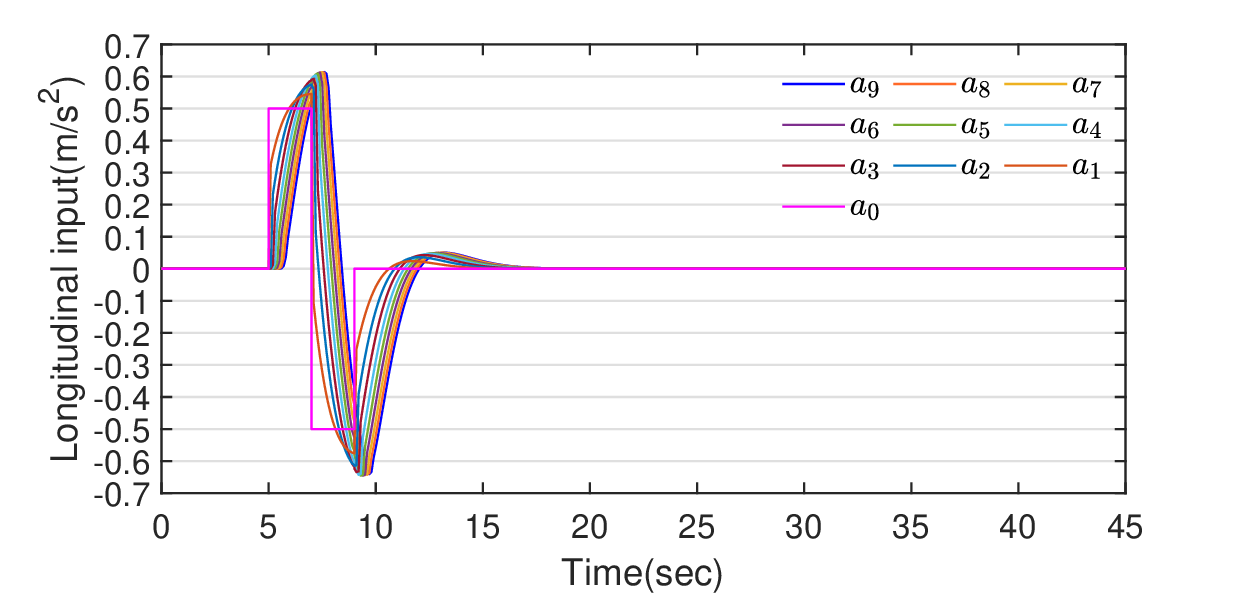}}
    \caption{\textcolor{black}{Longitudinal control: comparison of control input without communication delay and with delay $t_d=0.1$s for the proposed cooperative method. It can be seen that the proposed method is robust to such delay as there is only a slight increase in the peak of the control input. }}%
    \label{acce}%
\end{figure}

\begin{IEEEproof}
Using a Pade approximation technique as in \textit{Theorem \ref{thm2}}, the lateral error dynamics in frequency domain are
\begin{equation}\label{LL:52}
\begin{aligned}
\frac{{{{\hat e}_{i,3}}\left( s \right)}}{{{{\hat \omega }_{i - 1}}\left( s \right)}} &= \frac{{\! -\! \left( {s \!+\! \frac{{{v^*}}}{{{R^*}}}} \right)\mu \left( {2 \!-\! {t_d}s} \right) \!+\! \left( {\frac{{{v^*}}}{{{R^*}}} \!-\! {k_4}} \right)\left( {2 \!+\! {t_d}s} \right)}}{{\left( {{s^2} \!+\! \left( {{k_3} \!+\! {k_4}} \right)s \!+\! \frac{{{v^*}}}{{{R^*}}}{k_3}} \right)\left( {2 \!+\! {t_d}s} \right)}},\\
\frac{{{{\hat e}_{i,4}}\left( s \right)}}{{{{\hat \omega }_{i - 1}}\left( s \right)}} &= \frac{{\left( {2 \!+\! {t_d}s \!-\! \mu \left( {2 \!-\! {t_d}s} \right)} \right)s \!+\! {k_3}\left( {2 \!+\! {t_d}s} \right)}}{{\left( {{s^2} \!+\! \left( {{k_3} \!+\! {k_4}} \right)s \!+\! \frac{{{v^*}}}{{{R^*}}}{k_3}} \right)\left( {2 \!+\! {t_d}s} \right)}},
\end{aligned}
\end{equation}
and, for the preceding vehicle, % is calculated as
\begin{equation}\label{LL:53}
\begin{aligned}
\frac{{{{\hat e}_{i - 1,3}}\left( s \right)}}{{{{\hat \omega }_{i - 1}}\left( s \right)}}\! &=\! \frac{{\left( {\frac{{{v^*}}}{{{R^*}}}\frac{1}{{\mu s}} \!-\! \frac{{{k_4}}}{{\mu s}}} \right)\left( {2 \!+\! {t_d}s} \right) \!-\! \left( {\frac{{{v^*}}}{{{R^*}}}\frac{1}{s} \!+\! 1} \right)\left( {2\! - \!{t_d}s} \right)}}{{s\left( {2 - {t_d}s} \right) + \left( {\frac{{{k_4}}}{\mu } + \frac{{{k_3}}}{{\mu s}}\frac{{{v^*}}}{{{R^*}}}} \right)\left( {2 + {t_d}s} \right)}},\\
\frac{{{{\hat e}_{i - 1,4}}\left( s \right)}}{{{{\hat \omega }_{i - 1}}\left( s \right)}}\! &=\! \frac{{\left( {\frac{{{k_3}}}{{\mu s}} + \frac{1}{\mu }} \right)\left( {2 + {t_d}s} \right) - \left( {2 - {t_d}s} \right)}}{{s\left( {2 - {t_d}s} \right) + \left( {\frac{{{k_4}}}{\mu } + \frac{{{k_3}}}{{\mu s}}\frac{{{v^*}}}{{{R^*}}}} \right)\left( {2 + {t_d}s} \right)}}.
\end{aligned}
\end{equation}
Hence, the complementary sensitivity function in \eqref{LL:09} can be obtained based on \eqref{LL:52} and \eqref{LL:53} as
\begin{equation}\label{LL:54}
\begin{aligned}
\left|G_{i-1,i}(s)\right| &=\left| {\frac{{{{\hat e}_{i,3}}\left( s \right)}}{{{{\hat e}_{i - 1,3}}\left( s \right)}}} \right| = \left| {\frac{{{{\hat e}_{i,4}}\left( s \right)}}{{{{\hat e}_{i - 1,4}}\left( s \right)}}} \right| \\
% &= \left| {\mu s\frac{{s\left( {2 \!-\! {t_d}s} \right) \!+\! \left( {\frac{{{k_4}}}{\mu }\! +\! \frac{{{k_3}}}{{\mu s}}\frac{{{v^*}}}{{{R^*}}}} \right)\left( {2\! +\! {t_d}s} \right)}}{{\left( {{s^2} \!+\! \left( {{k_3} \!+\! {k_4}} \right)s \!+\! \frac{{{v^*}}}{{{R^*}}}{k_3}} \right)\left( {2 \!+\! {t_d}s} \right)}}} \right|,\\
&=\left| {\frac{{{f_0}{s^3} + {f_1}{s^2} + {f_2}s + {f_3}}}{{{g_0}{s^3} + {g_1}{s^2} + {g_2}s + {g_3}}}} \right|.
\end{aligned}
\end{equation}
with coefficients as
$f_0=- {t_d}\mu $,
$f_1=2\mu + {t_d}{k_4}$,
$f_2=2{k_4}+ {k_3}\frac{v^*}{R^*}{t_d}$,
$f_3=2{k_3}\frac{v^*}{R^*}$; and
$g_0={t_d}$,
$g_1=2+ {t_d}( {k_3} + {k_4} )$,
$g_2=2( {k_3} + {k_4})+ {k_3}\frac{v^*}{R^*}{t_d} $,
$g_3=2{k_3}\frac{v^*}{R^*}$.
Sufficient condition for lateral SFSS are obtained when the coefficients of the sixth-order polynomial in $\omega$ are negative, which is
\begin{equation}\label{LL:55}
\begin{aligned}
&t_d^2\left( {{\mu ^2} - 1} \right) \le 0,\\
&t_d^2{k_3}\frac{{{v^*}}}{{{R^*}}}\left( {\mu  + 1} \right) + {t_d}\left( {4\mu {k_4} - {k_3}{k_4}} \right) + 2\left( {{\mu ^2} - 1} \right) \le 0,\\
&- 8{k_3}\frac{{{v^*}}}{{{R^*}}}\left( {\mu \! -\! 1} \right) \le 4k_3^2 + 8{k_4}{k_3}.
\end{aligned}
\end{equation}
These give the conditions in \eqref{LL:51}. The proof is completed.
\end{IEEEproof}
%\begin{table}[b]
%\begin{center}
%\setlength{\abovecaptionskip}{0cm}
%\caption{Control parameters in longitudinal controller \eqref{LL:11}.} \label{table1}
%\scalebox{0.6}[0.6]{\includegraphics{tabledelay.pdf}}
%\vspace{-10pt}
%\end{center}
%\end{table}
%
%\begin{table}[b]
%\begin{center}
%\setlength{\abovecaptionskip}{0cm}
%\caption{Control parameters in literature \cite{18Khatir}.} \label{table2}
%\scalebox{0.6}[0.6]{\includegraphics{f.pdf}}
%\vspace{-10pt}
%\end{center}
%\end{table}

\begin{figure}[tp]
    \centering
    \setlength{\abovecaptionskip}{-0.1cm}
    \subfigure[Following angle error $e_{i,3}$ without communication delay.]{%
    \label{figure7a}%
     \includegraphics[width=3.2in]{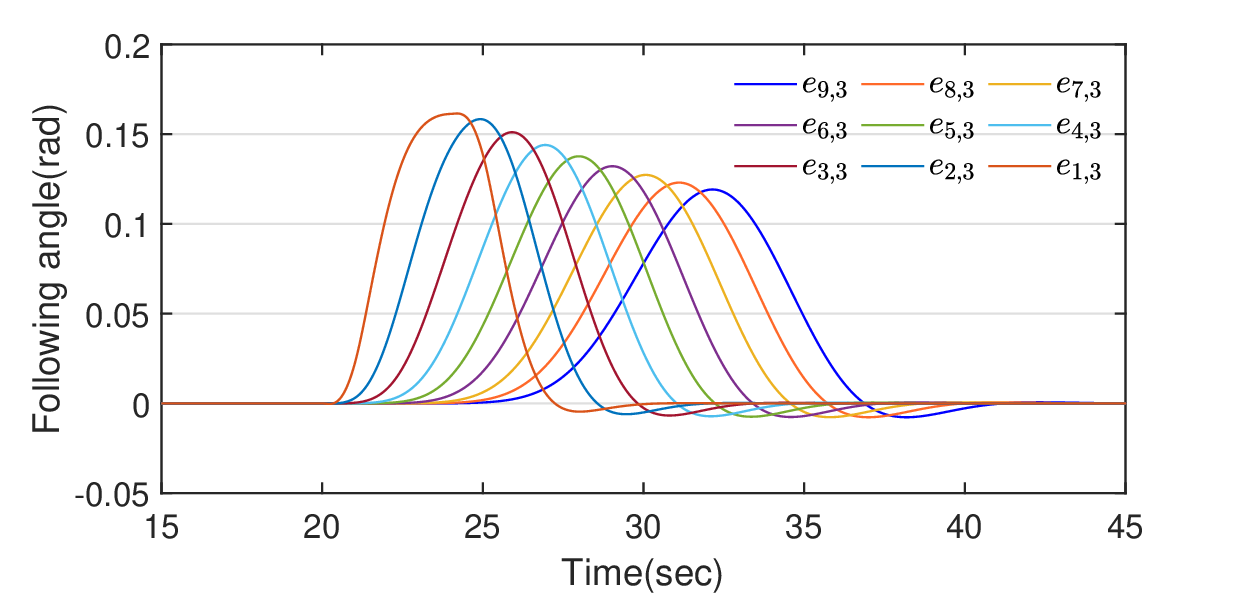}}
    \hspace{10pt}%
    \subfigure[Following angle error $e_{i,3}$ with delay $t_d=0.1s$.]{%
    \label{figure10a}%
     \includegraphics[width=3.2in]{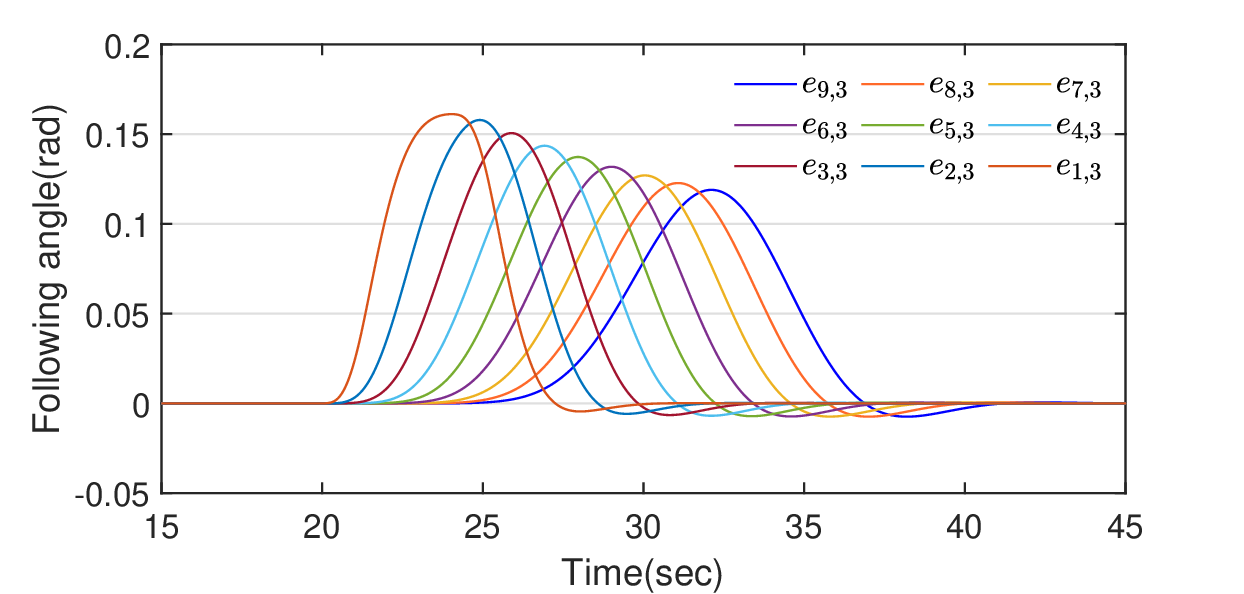}}
    \hspace{10pt}%
    \subfigure[Orientation angle error $e_{i,4}$ without communication delay.]{%
    \label{figure7b}%
    \includegraphics[width=3.2in]{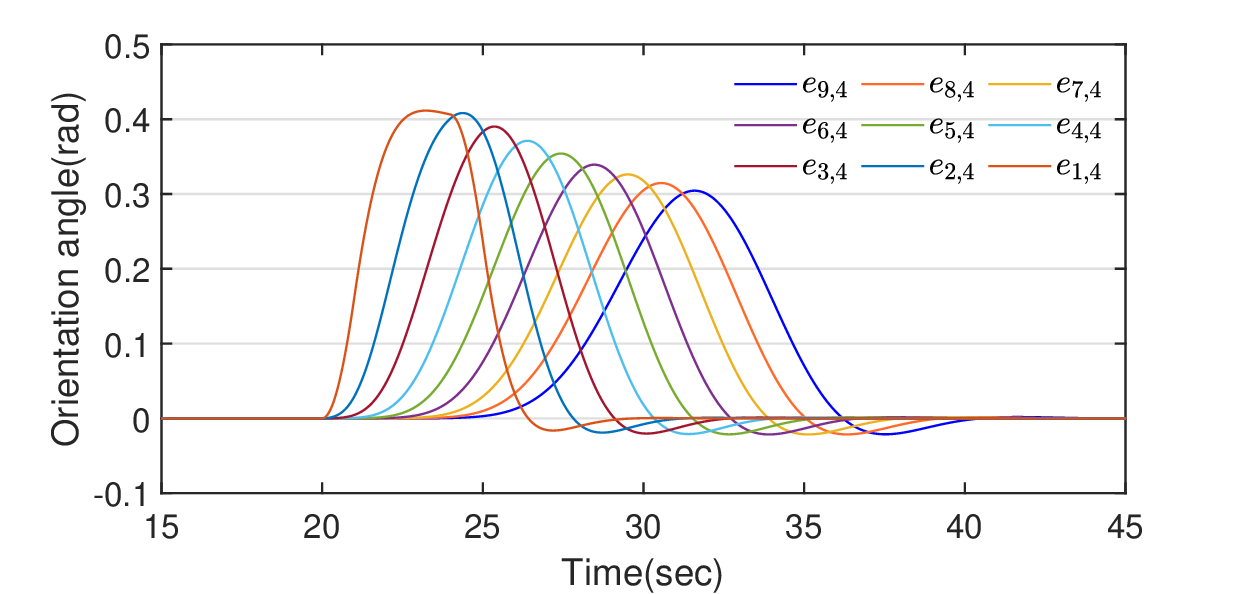}}
    \subfigure[Orientation angle error $e_{i,4}$ with delay $t_d=0.1s$.]{%
    \label{figure10b}%
    \includegraphics[width=3.2in]{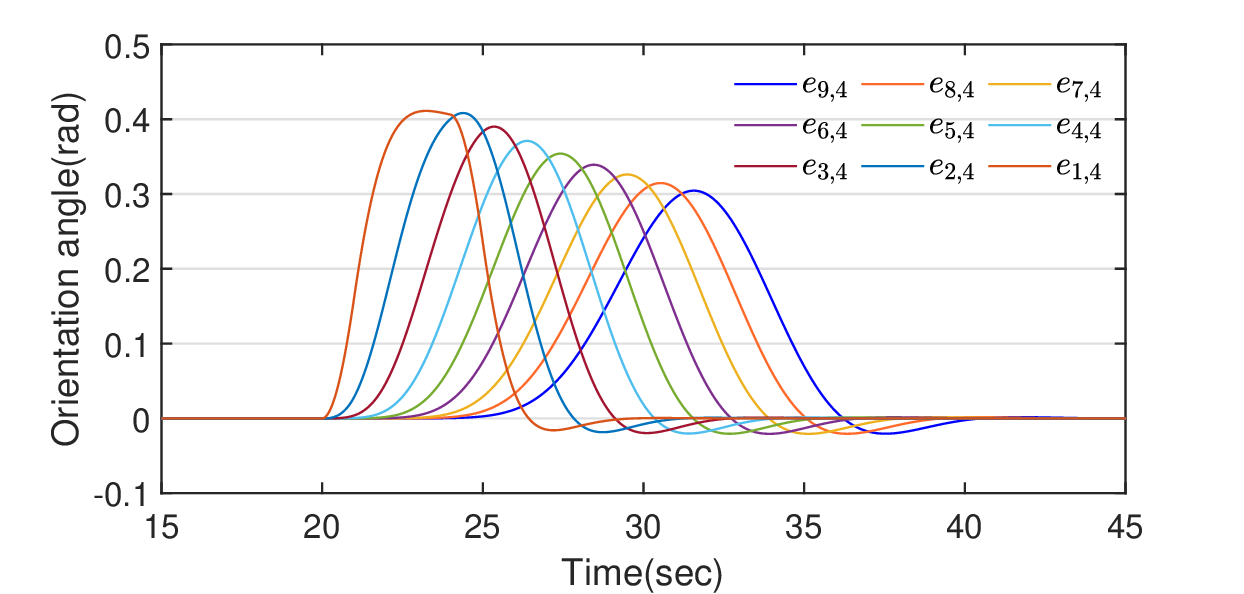}}
    \caption{\textcolor{black}{Lateral control: comparison of lateral errors without communication delay and with delay $t_d=0.1$s for the proposed cooperative method. It can be seen that the proposed method is robust to such delay as there is no noticeable effect on the error signals.}}%
    \label{e3e4}%
\end{figure}

\begin{figure}[tp]
    \centering
    \setlength{\abovecaptionskip}{-0.1cm}
    \subfigure[No communication delay.]{%
    \label{figure8a}%
    \includegraphics[width=3.2in]{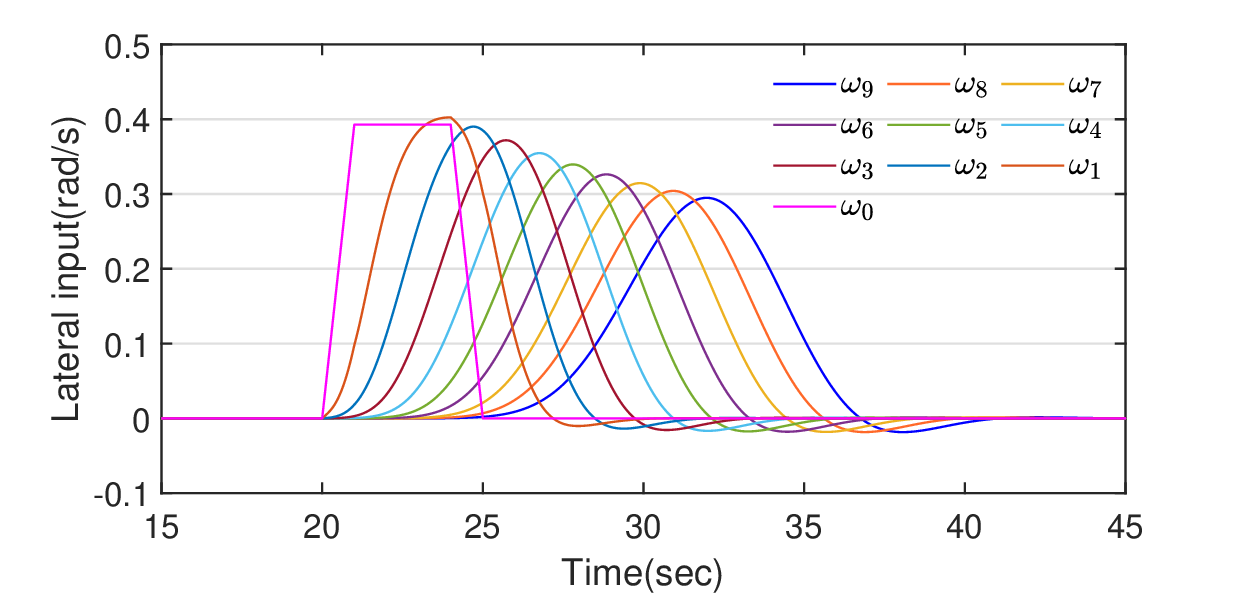}}
   % \hspace{10pt}%
    \subfigure[Communication delay $t_d=0.1s$.]{%
    \label{figure8b}%
    \includegraphics[width=3.2in]{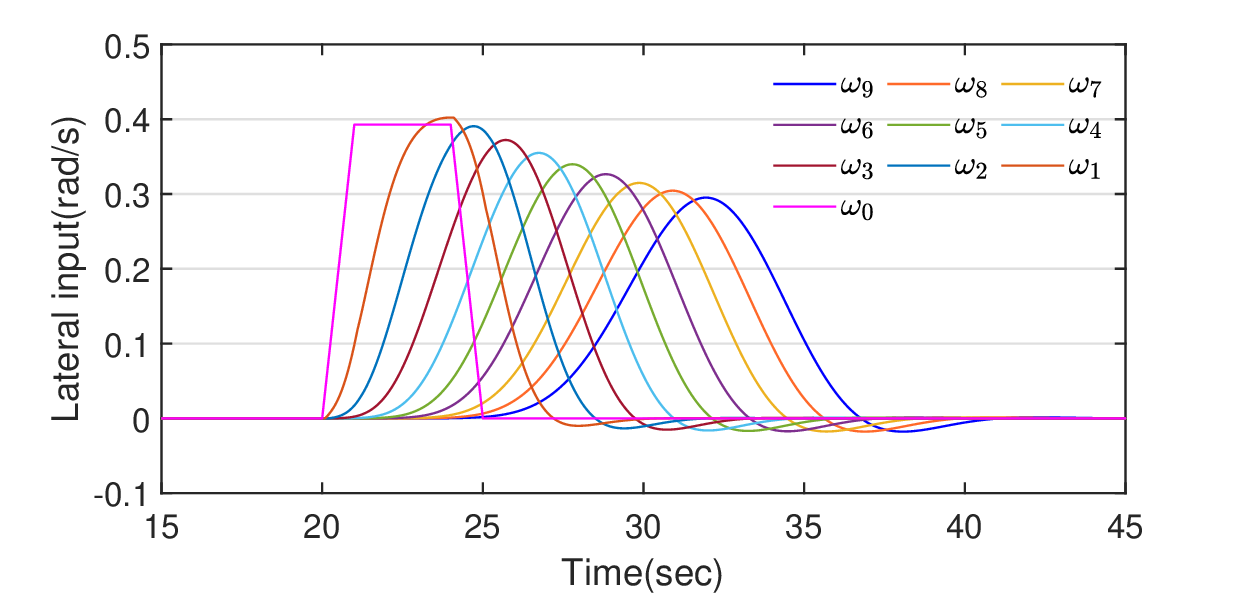}}
    \caption{\textcolor{black}{Lateral control: comparison of lateral control inputs without communication delay and with delay $t_d=0.1$s for the proposed cooperative method. It can be seen that the proposed method is robust to such delay as there is no noticeable effect on the input signals. }}%
    \label{omega}%
\end{figure}

\section{Numerical validation}\label{sec:05}
Let us consider a platoon with one leading and \textcolor{black}{nine} following vehicles,  with equilibria $v^*=10 {m \mathord{\left/ {\vphantom {m {{s^2}}}} \right.\kern-\nulldelimiterspace} {{s}}}$ and $R^*=10 m$. The leader behavior is as follows:
 \begin{equation*}
\textcolor{black}{u_0(t)= \left\{ \begin{aligned}
		&0,\ \ \ \ \ \ \ \ 0s \le t < 5s,\\
	    &0.5,\ \ \ \ \ \  5s \le t \le 7s,\\
		&\!-0.5,\ \ \  7s \le t \le 9s,\\
		&0, \ \ \ \ \ \ \ \ \ \ \ \ \ \ \  t > 9s.
	\end{aligned}
	\right.}
\end{equation*}
For the lateral maneuver, we simulate the vehicle platoon to turn by $90^\circ$ similar to \cite{18Khatir}. The leader behavior is as follows:
\begin{equation*}
\omega_0(t)= \left\{ \begin{aligned}
		&0,\ \ \ \ \ \ \ \ \ \ \ \qquad \qquad 0s \le t < 20s,\\
		&\frac{\pi }{8}(t-20),\ \ \ \ \ \ \qquad 20s \le t \le 21s,\\
		&\frac{\pi }{8},\ \ \ \ \ \ \ \ \ \ \qquad \qquad21s \le t \le 24s,\\
		&\frac{\pi }{8}-\frac{\pi }{8}(t-24),\ \ \ \ \ \ 24s \le t \le 25s,\\
		&0, \ \ \ \ \ \ \ \ \ \ \ \qquad \qquad t > 25s.
	\end{aligned}
	\right.
\end{equation*}

\subsection{Results of proposed cooperative control and comparisons with state-of-the-art decentralized control}\label{sec:05a}
\vspace{-5pt}
For the proposed longitudinal controller, we design the control gains in \eqref{LL:11} as $\alpha_i-\alpha_{i-1}=-0.1$ with $\alpha_1=1$, $\gamma_i=\gamma_{i-1}=0.5$, $i=2, \dots, 9$; $\beta_i$ calculated from \eqref{LL:18} according to the former two parameters with $\beta_1=1.6$. \textcolor{black}{The resulting non-identical control gains are
\begin{equation}
\begin{aligned} \nonumber
\alpha_i&=\{1,0.9,0.8,0.7,0.6,0.5,0.4,0.3,0.2\},\\
\beta_i&=\{1.6,2,2.4,2.8,3.2,3.4,3.6,3.8,4\}.
\end{aligned}
\end{equation}}Under the proposed controller \eqref{LL:11} with the above parameters, the evolution of relative distance $R_i$ and absolute velocity $v_i$ are shown in Fig. \ref{figure3b} and Fig. \ref{figure4b}. For comparisons with the state of the art, we reproduce the same conditions using the decentralized control in \cite{18Khatir} with control gains as in \eqref{23}.

\begin{figure}[tp]
    \centering
    \setlength{\abovecaptionskip}{-0.1cm}
    \subfigure[\textcolor{black}{Longitudinal control input.}]{%
    \label{figure18a}%
    \includegraphics[width=3.2in]{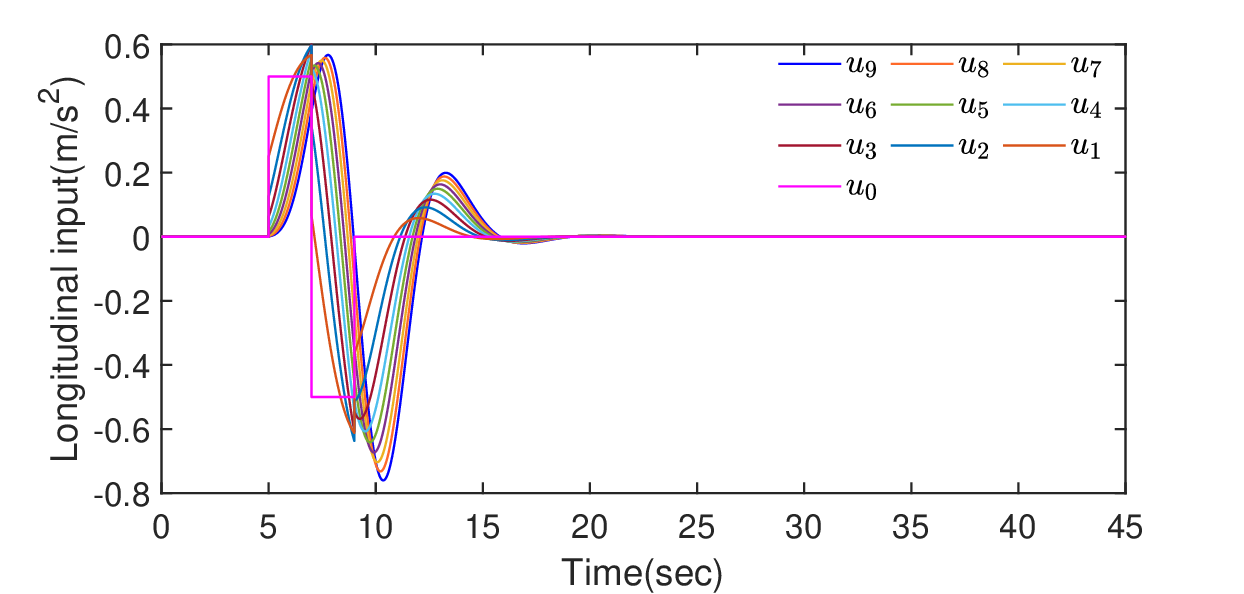}}
    \hspace{10pt}%
%    \subfigure[\textcolor{black}{Longitudinal acceleration.}]{%
%    \label{figure18b}%
    %\includegraphics[width=3in]{thirda.eps}}
    \subfigure[\textcolor{black}{Distance error $e_{i,1}$.}]{%
    \label{figure18c}%
    \includegraphics[width=3.2in]{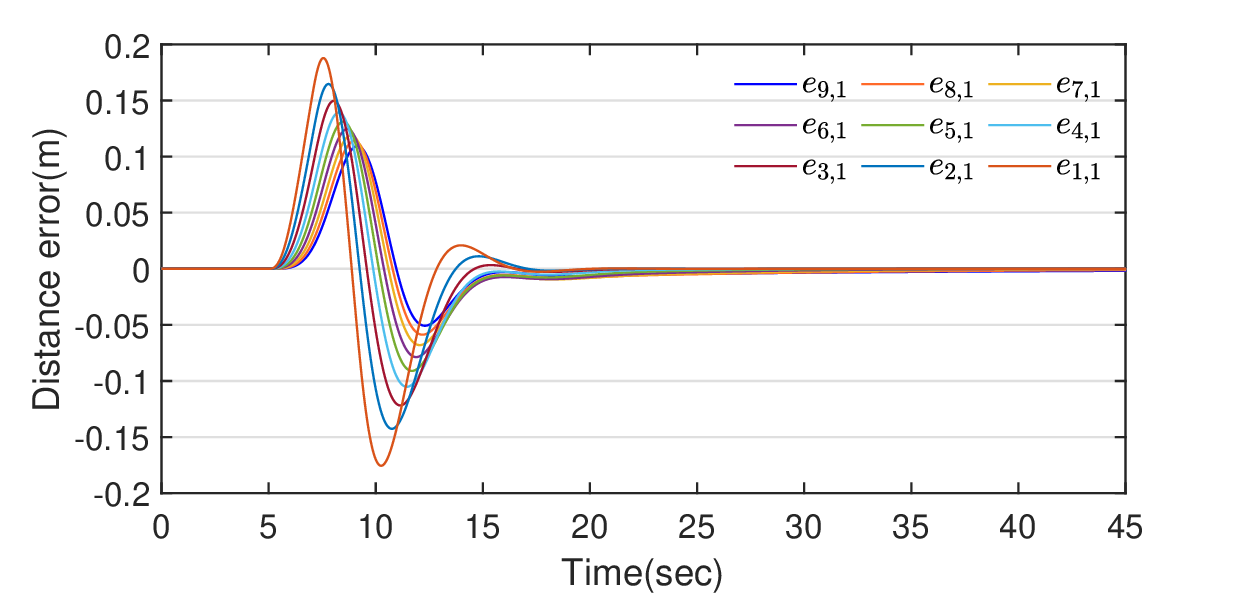}}
    \subfigure[\textcolor{black}{Velocity error $e_{i,2}$.}]{%
    \label{figure18d}%
    \includegraphics[width=3.2in]{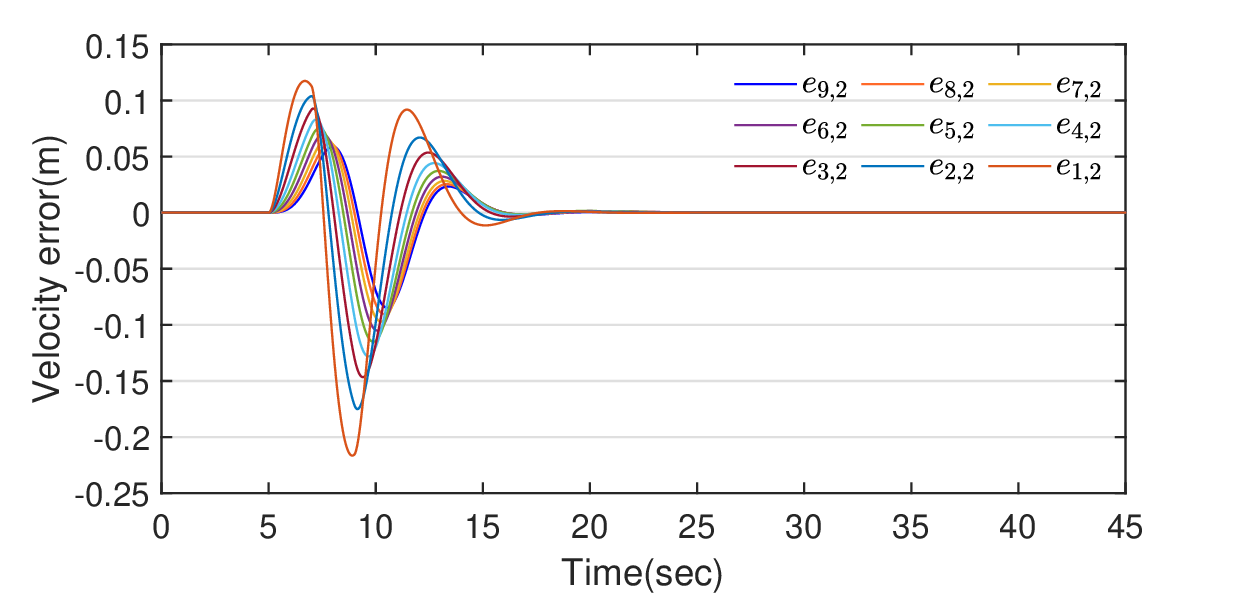}}
    \caption{\textcolor{black}{Longitudinal control with actuation time lag $\tau=0.5$s for the proposed cooperative method. It can be seen that the proposed method is robust to such time lag as there is string stability properties occur similar to Figs. \ref{e1e2} and \ref{acce}. }}%
    \label{third}%
\end{figure}

%\begin{figure}[tp]
%  \centering
%  \setlength{\abovecaptionskip}{-0.1cm}
%  \includegraphics[width=2.6in]{newlateralinput0729.eps}\\
%  \caption{Lateral input $\omega_i$ for the vehicle platoons.}\label{figure8}
%\end{figure}

%\begin{figure}[tp]
%  \centering
%  \setlength{\abovecaptionskip}{-0.1cm}
%  \includegraphics[width=2.8in]{acce.eps}
%  \caption{Longitudinal input $a_i$ for the vehicle platoons.}\label{figure5}
%\end{figure}

%\begin{figure}[tp]
%  \centering
%  \setlength{\abovecaptionskip}{-0.1cm}
%  \includegraphics[width=2.6in]{delaya.eps}
%  \caption{Longitudinal input $a_i$  with time-delay $t_d=0.1s$.}\label{figure10}
%\end{figure}

%\begin{figure}[tp]
%  \centering
%  \setlength{\abovecaptionskip}{-0.1cm}
%  \includegraphics[width=2.45in]{delayv.eps}
%  \caption{Longitudinal velocity $v_i$  with time-delay $t_d=0.1s$.}\label{figure11}
%\end{figure}

%\begin{figure}[tp]
%    \centering
%    \setlength{\abovecaptionskip}{-0.1cm}
%    \subfigure[Following angle $e_{i,3}$ with time-delay $t_d=0.1s$.]{%
%    \label{figure10a}%
%    %\includegraphics[width=2.5in]{delaye3.eps}}
%     \includegraphics[width=2.5in]{delaye310.eps}}
%    \hspace{10pt}%
%    \subfigure[Orientation angle $e_{i,4}$ with time-delay $t_d=0.1s$.]{%
%    \label{figure10b}%
%    %\includegraphics[width=2.5in]{delaye4.eps}}
%    \includegraphics[width=2.5in]{delaye410.eps}}
%    \caption{Lateral errors for the vehicle platoons with time-delay.}%
%    \label{delaye3e4}%
%\end{figure}

The comparative results are arranged in Fig. \ref{figure3a} and Fig. \ref{figure4a}. It is clear that the advantages of the proposed cooperative implementation are: faster convergence, less fluctuations and stronger string stability performance. To validate the string stability properties of the proposed controller, the relative distance error $e_{i,1}$ and relative velocity error $e_{i,2}$ are reported in Figs. \ref{figure6a}, \ref{figure6b}, from which the SFSS property is verified as the error peak decreases with the vehicle index. The evolution of the control input is shown in Fig. \ref{figure5a}.

For the proposed lateral controller \eqref{LL:29}, the parameters can be identical for all vehicles and are chosen as $k_3=2$, $k_4=0.1$ and $\mu=0.1$. \textcolor{black}{Stability and string stability are verified by looking at the following angle error $e_{i,3}$ in Fig. \ref{figure7a} and the orientation error $e_{i,4}$ in Fig. \ref{figure7b}: both errors  approach to zero, with the error peak decreasing with the vehicle index.} The lateral control input $\omega_i$ is depicted in Fig. \ref{figure8a}.

\subsection{Validation of robustness to delays}\label{sec:05b}
%In this part, considering the practical implementation, we
We now introduce a communication delay $t_d=0.1s$ in the longitudinal and lateral  terms $a_{i-1}\left( {t - {t_d}} \right)$ and $\omega_{i-1}\left( {t - {t_d}} \right)$. For the time-delayed longitudinal control, one can verify that all conditions in \eqref{LL:34} hold using the same control gains in Sect. \ref{sec:05a}. Therefore, these gains are robust to the communication delay $t_d=0.1s$.  \textcolor{black}{This is verified by the evolution of the longitudinal errors in Figs. \ref{figure9a}, \ref{figure9b}, by the evolution of the lateral errors in Figs. \ref{figure10a}, \ref{figure10b}, and by} the evolution of longitudinal control input in Fig. \ref{figure5b}. \textcolor{black}{The results have been organized so as to compare the no-delay scenario with the delayed scenario. Remarkably, the performance of the delayed scenario is nearly the same as the no-delay scenario. For longitudinal control, by comparing Figs. \ref{figure6a}, \ref{figure6b} (without delay) with Figs. \ref{figure9a}, \ref{figure9b} (with delay), there is an almost unnoticeable increase in the signal peaks, and it is evident that SFSS is maintained even with communication delay. Similar robustness considerations hold for lateral control.}

\textcolor{black}{To further verify robustness to actuation time lags, we perform the same scenario for the third-order dynamics \eqref{LL:24} with lag $\tau=0.5s$. The control gains are designed with the same $\alpha_i,\beta_i$ and $\gamma_i$ in Sect. \ref{sec:05a}, and $\lambda_i=0.5$ which satisfy the conditions in \eqref{new3}. Fig. \ref{figure18a} gives the control input of the vehicle platoon. Figs. \ref{figure18c}, \ref{figure18d} show that SFSS is guaranteed as distance and velocity errors both exhibit peak attenuation. Therefore, robustness to actuation lag is validated.}

\begin{figure}[tp]
  \centering
  \setlength{\abovecaptionskip}{-0.1cm}
  \includegraphics[width=1.8in]{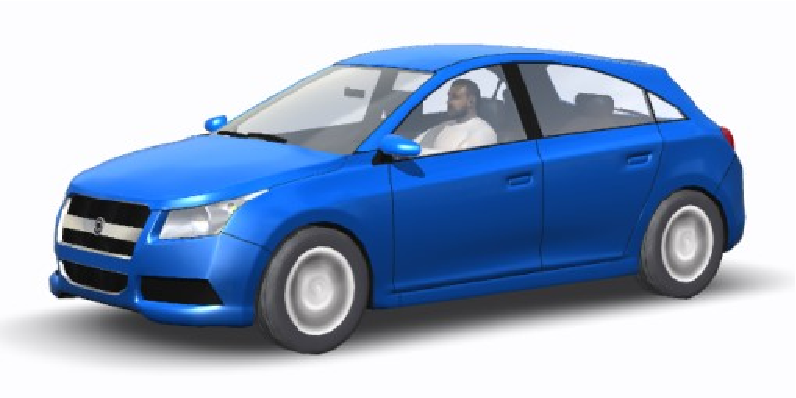}
  \caption{C-Class hatchback model in CarSim.}\label{figurecar}
\end{figure}
\begin{figure}[tp]
  \centering
  \setlength{\abovecaptionskip}{-0.1cm}
  \includegraphics[width=3.4in]{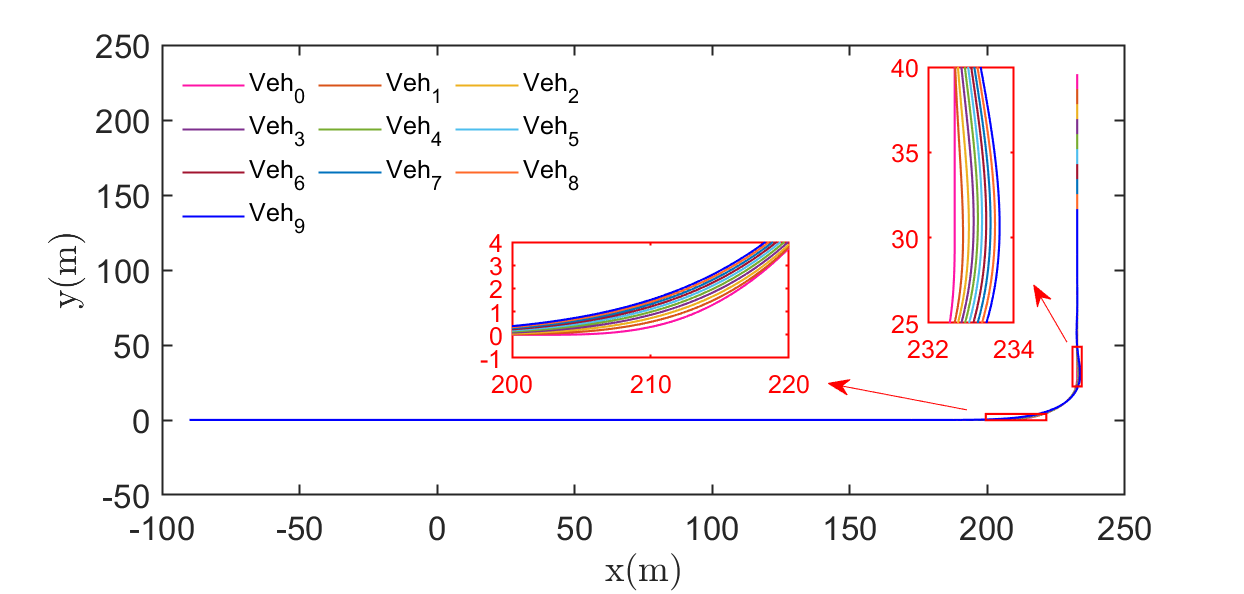}
  \caption{\textcolor{black}{CarSim results: Trajectories of the vehicle platoon in the $x$-$y$ plane. The driving scenario comprises a straight road followed by a $90^o $ turn.}}\label{figure13}
\end{figure}
\begin{figure}[tp]
  \centering
  \setlength{\abovecaptionskip}{0.1cm}
  \includegraphics[width=3.4in]{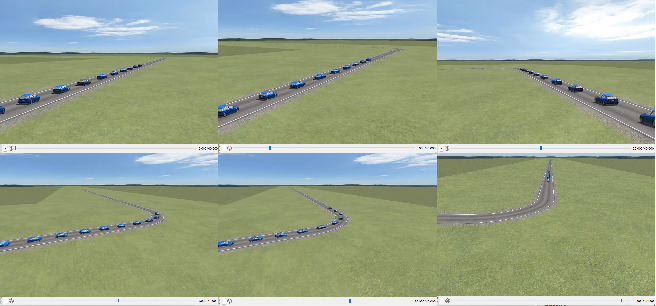}
  \caption{\textcolor{black}{CarSim results: Snapshots of the vehicle platoon at times $t = 0s; 8s; 20s; 23 s; 26s; 38s$. It can be seen that the platoon keeps the formation in the straight road and then turns while keeping the formation.}}\label{figure12}
\end{figure}

\begin{figure}[tp]
  \centering
  \setlength{\abovecaptionskip}{-0.1cm}
  \includegraphics[width=3in]{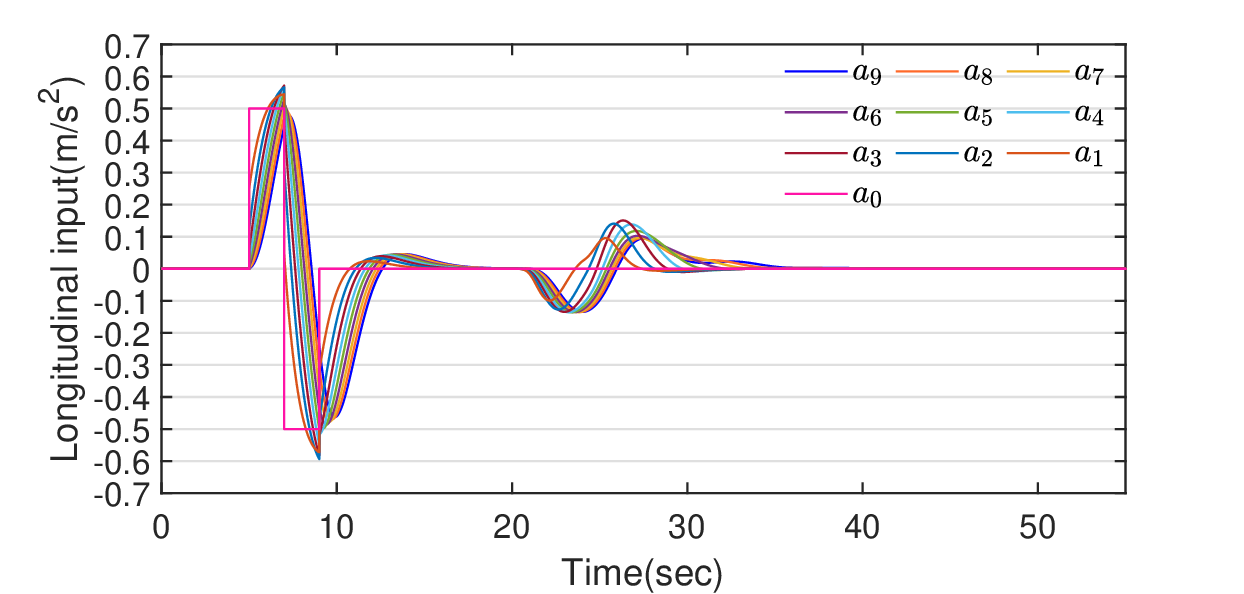}
  \caption{\textcolor{black}{CarSim results: Longitudinal input control input. The turning phase begins at around 20s, which requires the vehicles to slightly decelerate and then accelerate again.}}\label{figure14}
\end{figure}

\begin{figure}[tp]
  \centering
  \setlength{\abovecaptionskip}{-0.1cm}
  \includegraphics[width=3.2in]{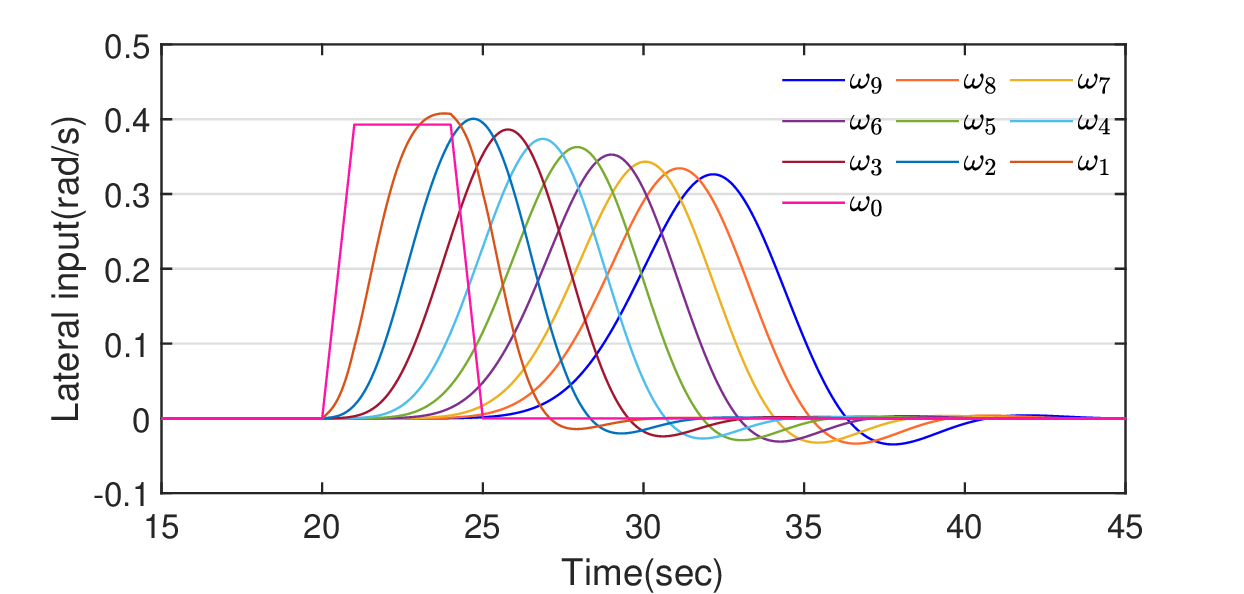}
  \caption{\textcolor{black}{CarSim results: Lateral control input.}}\label{figure15}
\end{figure}

\begin{figure}[t]
    \centering
    \setlength{\abovecaptionskip}{-0.1cm}
    \subfigure[Distance error $e_{i,1}$.]{%
    \label{figure16a}%
    \includegraphics[width=3.2in]{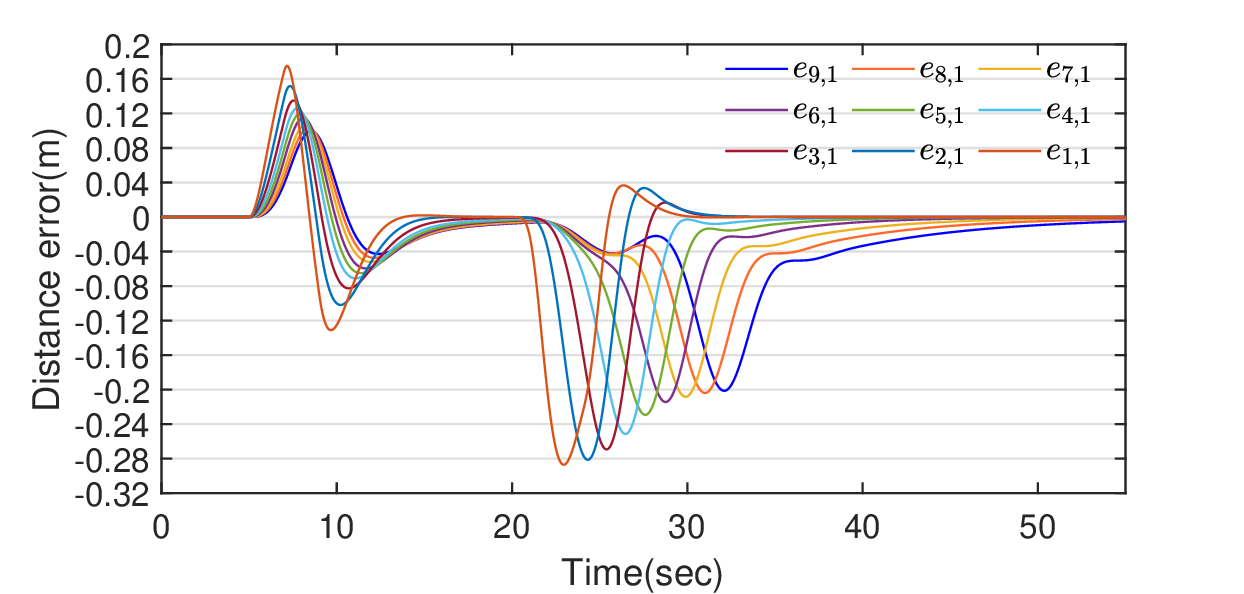}}
    \hspace{10pt}%
    \subfigure[Velocity error $e_{i,2}$.]{%
    \label{figure16b}%
    \includegraphics[width=3.2in]{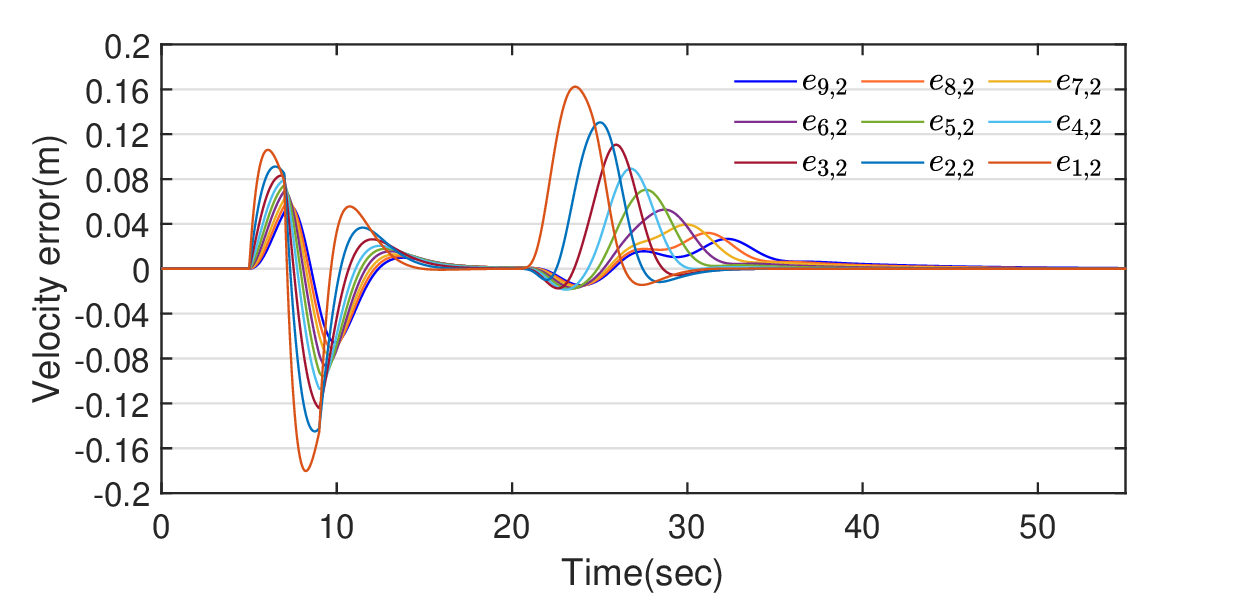}}
    \caption{\textcolor{black}{CarSim results: Longitudinal errors. Note that string stability is attained both during the straight road phase and during the turning phase. }}%
    \label{figure16}%
\end{figure}

\begin{figure}[t]
    \centering
    \setlength{\abovecaptionskip}{-0.1cm}
    \subfigure[Following angle error $e_{i,3}$.]{%
    \label{figure17a}%
    \includegraphics[width=3.2in]{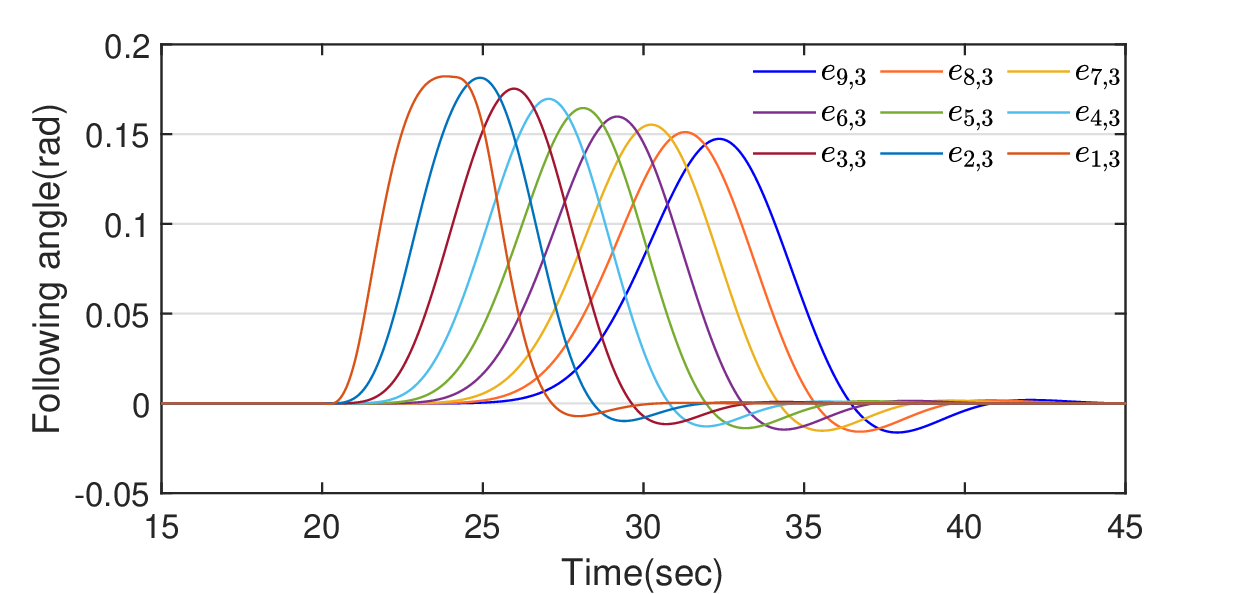}}
    \hspace{10pt}%
    \subfigure[Orientation angle error $e_{i,4}$.]{%
    \label{figure17b}%
    \includegraphics[width=3.2in]{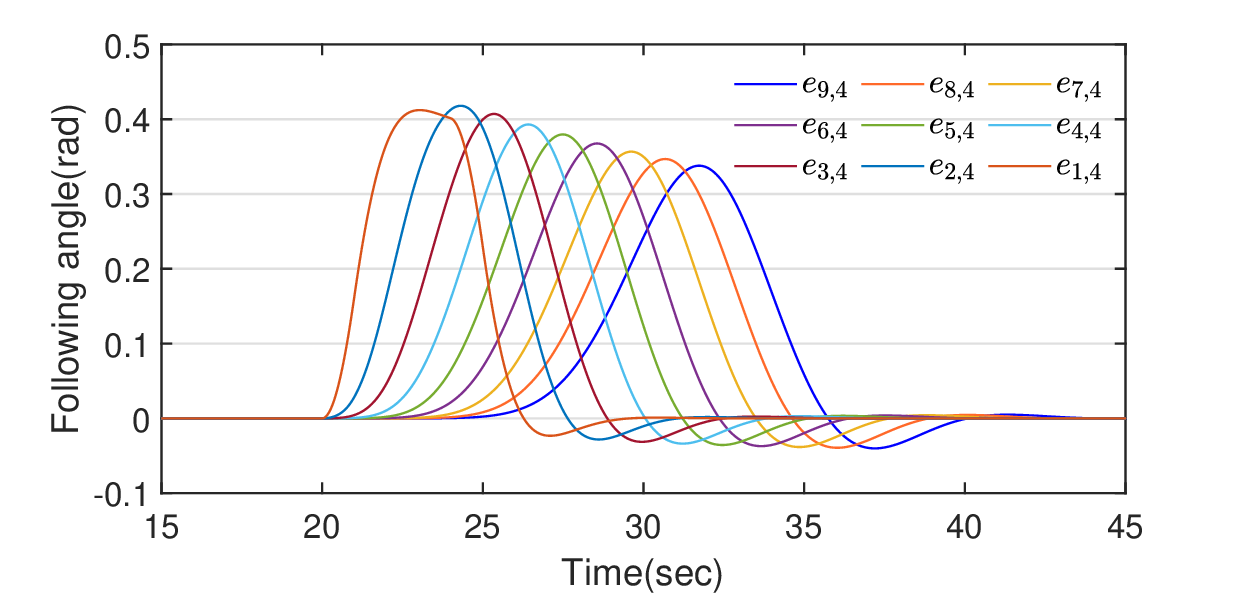}}
    \caption{\textcolor{black}{CarSim results: Lateral errors. Note that the string stability properties are analogous to those reported in Fig. \ref{e3e4}.}}%
    \label{figure17}%
\end{figure}

\subsection{Validation in CarSim mechanical simulation}
To show that the control laws \eqref{LL:11} and \eqref{LL:29} can be effective in a realistic scenario, we perform numerical tests in the software CarSim. The simulations are developed under aerodynamics of the C-Class hatchback in Fig.\ref{figurecar} whose wheelbase is $l=2.91m$ and friction coefficient is set to 0.85. CarSim can simulate engine and driveline dynamics that are not included in the control design, thus validating the methodology in a more realistic setting (engine parameters and other vehicle parameters in CarSim are the default parameters of C-class hatchback and are not reported for lack of space). To illustrate the simulation, we capture some snapshots of the platoons' motion video in Fig. \ref{figure12} at time $t=0s,7s,20s,22s,26s,35s$. The trajectory evolution in the $x$-$y$ plane is in Fig. \ref{figure13}, showing that the platoons can track the leading vehicle while turning. The longitudinal and lateral control inputs are in Fig.\ref{figure14} and Fig. \ref{figure15} respectively.
Fig. \ref{figure16a}-\ref{figure16b} shows that the relative distance and velocity errors maintain string stability while turning. % under the nonlinear coupling of the model.
Moreover, string stability is guaranteed for both lateral errors, as seen in Fig. \ref{figure17a}-\ref{figure17b}.

\ifCLASSOPTIONcaptionsoff
  \newpage
\fi

\section{Conclusions and future work}\label{sec:06}
This article considered %cooperative control for
mesh stability in vehicle platoons, i.e. ensuring string stability specifications in both longitudinal and lateral sense. Cooperative  control via vehicle-to-vehicle communication was the main feature of this work, as compared to available decentralized results only relying on on-board sensing. For longitudinal control, a non-identical cooperative controller was proposed  which extends and improves the decentralized controller in the literature by guaranteeing strong string stability. For lateral control, an identical cooperative controller was proposed which also guarantees lateral strong string stability. Robustness to \textcolor{black}{actuation time lags and} vehicle-to-vehicle communication delays have been analyzed.

Future work may involve studying the effect of parameter uncertainty on string stability. \textcolor{black}{Extending the proposed approach in a bidirectional sense so as to handle merging/splitting maneuvers is also of interest.} Moreover, while this work has adopted a frequency-domain mesh stability perspective, analyzing mesh stability based on Lyapunov theory is a promising research to address some nonlinearities. %. Nonlinearities can also include saturations due to actuator limitations, which can be potentially addressed using sector-bounded nonlinearities.
\bibliographystyle{IEEEtran}
\bibliography{sample}
\end{document}